\documentclass[11pt]{article}
\usepackage[a4paper,margin=1in]{geometry}
\usepackage{amsmath,amssymb,amsfonts,amsthm,mathtools}
\usepackage{mathrsfs}
\usepackage{bm}
\usepackage{tikz}
\usepackage{tikz-cd}
\usepackage{enumitem}
\usepackage{hyperref}
\hypersetup{colorlinks=true,linkcolor=blue,citecolor=blue,urlcolor=blue}
\usepackage{booktabs}
\usepackage{simpler-wick}
\usepackage{graphicx}
\usepackage{float} 
\usepackage{cite}
\usetikzlibrary{arrows.meta,positioning,decorations.pathmorphing}

\newcommand{\Fock}{\mathcal{F}}

\newtheorem{theorem}{Theorem}[section]
\newtheorem{lemma}[theorem]{Lemma}
\newtheorem{corollary}[theorem]{Corollary}
\newtheorem{proposition}[theorem]{Proposition}
\newtheorem{definition}[theorem]{Definition}
\newtheorem{example}[theorem]{Example}
\newtheorem{remark}[theorem]{Remark}

\newcommand{\C}{\mathbb{C}}

\newcommand{\pp}{\partial}
\newcommand{\ad}{\operatorname{ad}}
\newcommand{\End}{\operatorname{End}}
\newcommand{\Hom}{\operatorname{Hom}}

\newcommand{\Hilb}{\mathrm{Hilb}}

\newcommand{\FHilb}{\Fock_{Hilb}}

\newcommand{\Lam}{\Lambda}
\newcommand{\NO}[1]{:\!#1\!:}

\newcommand{\EHilb}{E_{1,\Hilb}}
\newcommand{\WHilbOne}{W_{\Hilb}^{(1)}}
\newcommand{\WHilbn}{W_{\Hilb}^{(n)}}
\newcommand{\WHilbTwo}{W_{\Hilb}^{(2)}}

\newcommand{\Z}{\mathbb Z}

\begin{document}

	\begin{center}
		{\Large\bf Geometric realization of $W$-operators}\\[6pt]
		{\large
			Lu-Yao Wang$^{a,b}$ \footnote{wangluyao@bimsa.cn}  }
		\vskip .2in
		
		$^a${\em Beijing Institute of Mathematical Sciences and Applications, Beijing, 101408, China}\\
		$^b${\em Yau Mathematical Sciences Center, Tsinghua University, Beijing, 100084, China} \\
		
	\end{center}
	
	\begin{abstract}   
		Certain integrable hierarchies appearing in random matrix theory, enumerative geometry, and conformal field theory are governed by Virasoro/$W$-algebra constraints and their $W$-representations.
		Motivated by the Gaussian Hermitian $\beta$-ensemble and recent studies of superintegrable partition function hierarchies, we build an explicit bridge from symmetric group class algebras to bosonic Fock spaces and further to geometry.
		
		On the algebraic side, we decompose the transposition class sum into cut and join channels and recover the classical cut-and-join operator on the ring of symmetric functions.
		On the geometric side, we use the Grojnowski-Nakajima Fock space identification to realize the ladder operator $E_1=[W_0,p_1]$ as the Hecke correspondence on $\mathrm{Hilb}_n(\mathbb C^2)$, and we interpret the cubic generator $W_0$ as a normal ordered triple incidence correspondence.
		We then explain how the $\beta$-deformed cubic generator $W_0^{(\beta)}$ arises from the Ward identities/Virasoro constraints of the Gaussian $\beta$-ensemble via a background charge parametrization, clarifying its conformal field theoretic meaning.
		Finally, using the Grojnowski-Nakajima Heisenberg-Fock isomorphism $\Phi_{\mathrm{Hilb}}:\Lambda\xrightarrow{\sim}\bigoplus_{n\ge0}H_T^*(\Hilb^n(\mathbb C^2))$,
		we transport the resulting commutator hierarchy to Hilbert schemes, where $E_1$ is realised by the Hecke correspondence
		(adding one point) and the diagonal correction terms are computed by equivariant localization from the $T$-weights of
		the tangent bundle $T\Hilb^n(\mathbb C^2)$ and the tautological bundle $\mathcal V$.
		This provides a geometric realization framework that unifies $\beta$-deformed integrable structures and offers new tools for studying quiver gauge theory partition functions.
	\end{abstract}

	\tableofcontents

	\section{Introduction}
	
	\paragraph{Physical motivation.}\quad
	
	Matrix model partition functions are controlled by Ward identities that organize into Virasoro and,
	more generally, $W$-algebra constraints \cite{DFGZJ_2DGravityRandomMatrices,BouwknegtSchoutens_Wreview}.
	In the Hermitian Gaussian $\beta$-ensemble, these constraints can be written as differential operators
	in the time variables (or power sums) acting on $Z_G^{(\beta)}$
	\cite{Eynard_CountingSurfaces,DiFrancescoMathCFT}.
	Equivalently, many such partition functions admit $W$-representations of the form
	\[
	Z(\mathbf p)=\exp\!\bigl(\mathcal W(\mathbf p,\partial_{\mathbf p})\bigr)\cdot 1,
	\]
	where $\mathcal W$ is a finite (often cubic) combination of Heisenberg modes
	\cite{Alexandrov_CutJoin_GKM,MironovMorozov_WrepSurvey}.
	This perspective is closely tied to superintegrability (Schur/Jack expansions with factorized averages)
	in $\beta$-deformed hierarchies \cite{Macdonald95,Stanley_Jack}.
	
	A second physical input is the $\Omega$-background of four-dimensional $\mathcal N=2$ gauge theory,
	where the deformation parameters $\varepsilon_1,\varepsilon_2$ enter the instanton partition function and one
	commonly sets $\beta=\varepsilon_1/\varepsilon_2$ \cite{Nekrasov_Instantons,NekrasovOkounkov2003}.
	In AGT type correspondences, such partition functions are identified with conformal blocks of 2d CFTs
	\cite{AGT_2009}, and the free field realization introduces a background charge $Q$ in the improved stress tensor,
	shifting the central charge \cite{DiFrancescoMathCFT,FeiginFuchs_Virasoro}.
	
	From this standpoint, a key problem is:
	
	\emph{Is it possible to geometrically realize the generating $W$-operators (and the integrable hierarchies they generate)
		underlying $\beta$-ensembles and gauge theory partition functions, in a manner compatible with $\Omega$-background
		parameters and quiver gauge theory?}
	
	Nakajima quiver varieties provide a natural candidate geometric stage:
	they geometrically engineer instanton moduli spaces and carry canonical actions of Heisenberg/Yangian algebras
	\cite{Nakajima94,NakajimaAnnals,NakajimaLectures,MaulikOkounkov12,SchiffmannVasserot_EllipticHallHilb}.
	
	\paragraph{Mathematical motivation}\quad
	
	The classical cut-and-join operators provide an efficient way to encode the
	action of transpositions on conjugacy classes of the symmetric group and to
	translate Hurwitz-type counting problems into differential equations for
	generating series of symmetric functions \cite{GouldenJacksonTranspositions,
		GouldenJacksonKP,OkounkovRM,sun2019formula,OkounkovTodaHurwitz}. In particular, the celebrated Hurwitz operator
	is a cubic differential operator in the power sum variables whose eigenbasis is
	given by Schur functions, and whose spectral data control simple branched
	coverings of the Riemann sphere. From the point of view of symmetric
	functions, this operator acts on the algebra $\Lambda$ with respect to the
	standard basis of power sums and the classical inner product introduced by
	Macdonald \cite{Macdonald95}. 
	
	A fundamental insight due to Okounkov is that the Fock space formalism underlying cut-and-join operators and Hurwitz theory extends naturally to geometric curve counting and the integrable structures underpinning enumerative theories.
	Specifically, random partitions and Hurwitz generating series are realized as $\tau$-functions,
	with the associated hierarchies (KP/Toda and their variants) controlled by operator constraints on the bosonic Fock space \cite{OkounkovRM,OkounkovTodaHurwitz,OkounkovPandharipandeGW}
	This perspective provides a conceptual bridge linking symmetric function differential operators to geometric representation theory, where these operators are realized geometrically via correspondences and (quantum) multiplication by tautological classes.

	In parallel, the equivariant (co)homology of Hilbert schemes of points on the plane is endowed with a natural action of the Heisenberg algebra, as constructed by Grojnowski and Nakajima \cite{NakajimaLectures,NakajimaAnnals,Groj}.
	More precisely, the
	direct sum
	\[
	\mathcal F_{\Hilb}
	= \bigoplus_{n\ge 0} H_T^\ast\!\bigl(\Hilb^n(\C^2)\bigr)
	\]
	is naturally isomorphic to the bosonic Fock space of symmetric functions,
	and under this identification the Nakajima creation and annihilation operators
	realize the usual Heisenberg algebra acting on $\Lambda$. 
	
	Okounkov's work elevates this identification to the paradigm of ``geometric integrability''
	: (i) operator descriptions of (quantum) multiplication and curve counting on
	$\Hilb(\C^2)$ and related moduli spaces can be formulated directly in the Fock space,
	with spectral/eigenvalue data expressed in terms of partitions and box-contents
	\cite{OkounkovPandharipandeHilb,OkounkovPandharipandeGW}; and (ii) more generally, for
	Nakajima quiver varieties, the work of Maulik--Okounkov constructs stable envelopes
	and uses them to produce $R$-matrices and quantum-group actions (Yangian/quantum affine/toroidal),
	leading to quantum (difference/differential) equations whose solutions are encoded by
	the same fixed-point combinatorics \cite{MaulikOkounkov12}. 
	
	This Fock space representation has played a central role in geometric representation theory,
	in the study of Macdonald polynomials and in the realization of (quantum)
	toroidal algebras via the elliptic Hall algebra and its shuffle models
	\cite{SchiffmannVasserot_EllipticHallHilb,SchiffmannVasserot12}.

	\paragraph{Main results}\quad
	
	The purpose of this paper is to give an explicit, computable realization of the basic $W$-generator and its hierarchy on Hilbert schemes.
	Concretely, we give a self-contained and explicit bridge between
	these two viewpoints. 
	Starting from the group algebra $\C[S_n]$ we decompose
	the transposition class sum into cut and join channels and compute their images
	under the cycle-index map to the ring $\Lambda$ of symmetric functions,
	recovering the classical cut-and-join operator together with its diagonal
	correction term. We then introduce a cubic Heisenberg operator $W_0$ and a
	ladder construction which lift the symmetric group action to the equivariant
	(co)homology of the Hilbert schemes. Both the classical case $\beta=1$ and its
	$\beta$-deformation (Jack case) are treated uniformly.
	
	A key technical point is that, after a suitable normalization of the
	Heisenberg generators, normal ordering produces only $\beta$-dependent
	diagonal corrections, while the top-derivative (principal symbol)
	part of the resulting hierarchy of $W$-operators is independent of
	$\beta$.  This yields a clean graded structure compatible with the Euler
	operator, and shows that the underlying cut-and-join graph (the leading
	combinatorics of splitting and joining cycles) is universal across
	these deformations, with $\beta$ only modifying lower order and
	diagonal weights.
	From a broader perspective, our construction can be viewed as a geometric and representation theoretic interpretation of the classical Hurwitz cut-join formalism
	\cite{GouldenJacksonTranspositions,GouldenJacksonKP,OkounkovTodaHurwitz},
	suggesting applications to refined Hurwitz theory, curve counting, and
	integrable hierarchies.	Recently, $\beta$-deformed matrix models and their partition function
	hierarchies with $W$-representations have been systematically studied,
	see the WLZZ model in~\cite{wlzz22},
	where the superintegrability property and Schur/Jack character
	expansions are analyzed for a wide class of $(\beta\text{-deformed})$
	matrix models.

	Based on the above analysis, the main results of this paper can be summarized as follows.
	
	\begin{itemize}
		\item[(1)] \textbf{From transpositions to cut/join and a cubic Heisenberg operator.} 
		
		\begin{theorem}(Theorem~\ref{thm:cut-join})
			Under the characteristic map $\Phi_n$, the actions of
			$K^{(n)}_{\mathrm{cut}}$ and $K^{(n)}_{\mathrm{join}}$ on $\Lambda^{(n)}$ are
			given by the differential operators
			\begin{equation*}
				\Phi_n\circ K^{(n)}_{\mathrm{cut}}\circ\Phi_n^{-1}
				=
				\frac12\sum_{a,b\ge1}(a+b)\,p_a p_b\,\partial_{p_{a+b}},\quad
				\Phi_n\circ K^{(n)}_{\mathrm{join}}\circ\Phi_n^{-1}
				=
				\frac12\sum_{a,b\ge1}ab\,p_{a+b}\,\partial_{p_a}\partial_{p_b}.
			\end{equation*}
		\end{theorem}
		
		\item[(2)] \textbf{Ladder construction and Jucys-Murphy lifting.}
		(Lemma \ref{thm:centered_ladder} and \ref{thm:JM_lifting_centered})
		
		We introduce a layered block algebra and a raising element $E$ whose action is
		intertwined, under the cycle-index map, with the differential operator
		$	E_1 = [W_0,p_1] = \sum_{m\ge1} m\,p_{m+1}\partial_{p_m}.$
		This construction admits a group theoretic interpretation in terms of
		Jucys-Murphy elements and provides a convenient bridge between the
		symmetric group side and the geometric realization on Hilbert schemes.
		
		\item[(3)] \textbf{Geometric realization on Hilbert schemes and $\beta$-deformation.}	
		
		\begin{theorem}(Theorem~\ref{thm:Phi_Hilb_graded})
			There exists a unique graded linear isomorphism
			\[
			\Phi_{\Hilb}:\Lambda \xrightarrow{\ \sim\ }\Fock_{\Hilb},
			\quad
			\Phi_{\Hilb}=\bigoplus_{n\ge 0}\Phi_{\Hilb,n},\ \ \Phi_{\Hilb,n}:\Lambda^{(n)}\xrightarrow{\sim} H_T^*(\Hilb^n(X)),
			\]
			such that
			\[
			\Phi_{\Hilb}(1)=|0\rangle \in H_T^*(\Hilb^0(X))
			\quad\text{and}\quad
			\Phi_{\Hilb}\!\bigl(p_k\,f\bigr)=\frac1k\,\alpha_{-k}\,\Phi_{\Hilb}(f)\quad (k\ge 1,\ f\in\Lambda).
			\]
			Equivalently, $\Phi_{\Hilb}$ intertwines the Heisenberg actions with the normalization \eqref{eq:Heis_normalization_Lambda}:
			\[
			\Phi_{\Hilb}\circ (k\,p_k)=\alpha_{-k}\circ \Phi_{\Hilb},\qquad
			\Phi_{\Hilb}\circ \partial_{p_k}=\alpha_{k}\circ \Phi_{\Hilb}\qquad (k\ge 1).
			\]
		\end{theorem}
		
		\begin{theorem}(Theorem~\ref{thm:W0beta})
			Define the cubic operator on $ \Lambda=\C[p_1,p_2,\dots]$ by
			\begin{equation*}
				W_0^{(\beta)}
				=\frac{\beta}{2}\sum_{i,j\ge 1}\NO{a_{-i}a_{-j}a_{i+j}}
				+\frac{1}{2}\sum_{i,j\ge 1}\NO{a_{-(i+j)}a_i a_j}
				+\frac12\sum_{k\ge 1}\Bigl((1-\beta)(k-1)+2\beta N\Bigr)a_{-k}a_k,
			\end{equation*}
			where $a_{\pm k}$ are the Heisenberg operators (identified with \eqref{eq:Heis_norm}) and
			$\NO{\cdot}$ denotes normal ordering with respect to the Fock grading.
			Let $W_{0,\mathrm{geo}}^{(\beta)}$ be the operator on $\FHilb$ defined by the same
			normally ordered expression \eqref{eq:W0beta}, but interpreted as a correspondence operator
			on $\Hilb^n(\C^2)$ obtained by convolving Nakajima incidence correspondences (and for the
			diagonal term, by cup product with the canonical equivariant class that produces the
			normal ordering correction).
			Then
			\[
			\Phi_{\Hilb}\circ W_0^{(\beta)}\;=\;W_{0,\mathrm{geo}}^{(\beta)}\circ \Phi_{\Hilb}.
			\]
		\end{theorem}

		\item[(4)] \textbf{Hierarchy operator and geometric interpretation.}
		\begin{theorem}(Theorem~\ref{thm:WHilb})
			For all $n\ge2$ the operators $\WHilbn$ defined in
			Definiiton \ref{eq:Wgeo_def} satisfy the following properties.
			
			\begin{enumerate}
				\item \emph{Compatibility with the Fock isomorphism.}
				For all $n\ge1$, we have  the intertwining relation
				$	\Phi_{\Hilb}\circ W^{(n)} = \WHilbn\circ\Phi_{\Hilb}.
				$
				
				\item \emph{First non–trivial level $n=2$.}
				The operator $\WHilbTwo$ is the commutator of the degree-$1$
				Hamiltonian with the Hecke step:
				\[
				\WHilbTwo = [\,\WHilbOne,\EHilb\,].
				\]
				Geometrically, $\WHilbTwo$ is represented by the difference of the
				two composite correspondences
				\[
				\Hilb^n \xrightarrow{\ \EHilb\ } \Hilb^{n+1}
				\xrightarrow{\ \WHilbOne\ } \Hilb^{n+2},
				\quad
				\Hilb^n \xrightarrow{\ \WHilbOne\ } \Hilb^{n+1}
				\xrightarrow{\ \EHilb\ } \Hilb^{n+2}.
				\]
				
				\item \emph{Higher levels as iterated commutators.}
				For
				$n\ge2$ we have
				\begin{equation*}
					\WHilbn = \mathrm{ad}_{\WHilbOne}^{\,n-1}(\EHilb).
				\end{equation*}
				In particular, $\WHilbn$ maps $H^*_{T}(\Hilb^N)$ to
				$H^*_{T}(\Hilb^{N+n})$, and it is represented by an alternating sum
				over all ways of interlacing a single Hecke correspondence
				$E_{1,\Hilb}$ with $(n-1)$ copies of $W_{[2],\Hilb}$ (or, equivalently,
				$(n-1)$ copies of $\WHilbOne$) along flags of nested ideals
				$
				I_N \supset I_{N+1} \supset \cdots \supset I_{N+n}
				$
				of total length $n$.
			\end{enumerate}
		\end{theorem}

	\end{itemize}
	
	In conclusion, our discussion will give the relation between  Group algebra $\to$ symmetric functions  $\to$ Hilbert schemes (here we take the raising operator $E_1$ as an example):

	\[
	\begin{tikzcd}[column sep=large,row sep=large]
		Z_n\text{ (block/JM)} \arrow[r,"\Phi_n"] \arrow[d,"\text{``add one point''}"'] 
		& \Lambda^{(n)} \arrow[r,"\Phi^{(n)}_{\Hilb}"] \arrow[d,"E_1\ \text{or}\ E_i"] 
		& 	H_T^*(\Hilb^{n}(\C^2))\arrow[d,"C\ \text{or}\ C_i"]\\
		Z_{n+1} \arrow[r,"\Phi_{n+1}"'] 
		& \Lambda^{(n+1)} \arrow[r,"\Phi^{(n+1)}_{\Hilb}"'] 
		& 	\ H_T^*(\Hilb^{n+1}(\C^2)) \ .
	\end{tikzcd}
	\]
	\subsection*{Organization}
	In section~\ref{sec:ladder}, we construct the cut-and-join operator from the symmetric group side and relates the raising operator $E_1$ to Jucys-Murphy lifting.
	In section~\ref{sec:beta}, we places the Gaussian Hermitian $\beta$-ensemble at the center:  derive Virasoro constraints and package them into the cubic constraint $W_0^{(\beta)}$ together with its background charge interpretation. 
	Then, we realizes these operators geometrically on $\mathrm{Hilb}_n(\mathbb C^2)$ via Hecke and incidence correspondences and develops the exchange diagram calculus in section~\ref{sec:hilb}.
	Specially, we use the Grojnowski-Nakajima Heisenberg-Fock isomorphism $\Phi_{\mathrm{Hilb}}:\Lambda\xrightarrow{\sim}\bigoplus_{n\ge0}H_T^*(\Hilb^n(\mathbb C^2))$ to transport the hierarchy operator to Hilbert schemes.
	Last, we summarizes this work and lists physics motivated directions (superintegrability, BPS/CFT, and deformed integrable hierarchies).
	
	\newpage
	
	\section{Symmetric group center and the characteristic map}	\label{sec:cutjoin}
	
	\paragraph{Symmetric group center and the characteristic map}\quad
	
	For the symmetric group $S_n$, let $
	K^{(n)}_{[2]}=\sum_{1\le i<j\le n}(ij)$
	be the sum of all transpositions in the group algebra $\mathbb{C}[S_n]$.
	Its center is
	\[
	Z_n:=Z(\C[S_n])=\{z\in \C[S_n]\mid zx=xz,\ \forall x\in \C[S_n]\}.
	\]
	It has a standard basis indexed by partitions $\lambda\vdash n$: the \emph{class sums}
	\[
	C_\lambda := \sum_{\sigma\in \mathcal C_\lambda}\sigma,
	\]
	where $\mathcal C_\lambda\subset S_n$ is the conjugacy class of cycle type $\lambda$.

	Let $\Lambda=\C[p_1,p_2,\dots]$ be the algebra of symmetric functions in power sums, graded by
	\[
	\deg(p_k)=k,\qquad \Lambda=\bigoplus_{n\ge 0}\Lambda^{(n)}.
	\]
	Then $\dim_\C Z_n=\#\{\lambda\vdash n\}=\dim_\C \Lambda^{(n)}$.
	
	\begin{definition}[Characteristic/cycle index map \texorpdfstring{$\Phi_n$}{Phi-n}]
		Define a linear map
		\[
		\Phi_n: Z_n\longrightarrow \Lambda^{(n)}
		\]
		by sending the class sum basis to power sums:
		\begin{equation}\label{eq:Phi_def}
			\Phi_n(C_\lambda)=p_\lambda:=\prod_{i}p_{\lambda_i}\qquad(\lambda\vdash n).
		\end{equation}
	\end{definition}
	
	\begin{lemma}[Isomorphism]\label{lem:Phi_iso}
		The map $\Phi_n:Z_n\to \Lambda^{(n)}$ defined in \eqref{eq:Phi_def} is a vector space isomorphism.
	\end{lemma}
	\begin{proof}
		Both sides have a basis indexed by partitions $\lambda\vdash n$ and $\Phi_n$ maps one basis bijectively to the other.
	\end{proof}

	\paragraph{Cut/join from the transposition class sum}\quad	
	
	The action of $K^{(n)}$ splits into cut and join contributions
	\[
	K^{(n)} = K^{(n)}_{\mathrm{cut}}+K^{(n)}_{\mathrm{join}},
	\]
	where $K^{(n)}_{\mathrm{cut}}$ collects those transpositions that cut a cycle of a permutation into two cycles, while
	$K^{(n)}_{\mathrm{join}}$ collects those that join two cycles
	into one.  The following theorem gives the corresponding differential
	operators on $\Lambda$.
	
	\begin{theorem}[Cut/join operators]
		\label{thm:cut-join}
		Under the characteristic map $\Phi_n$, the actions of
		$K^{(n)}_{\mathrm{cut}}$ and $K^{(n)}_{\mathrm{join}}$ on $\Lambda$ are
		given by the differential operators
		\begin{align}
			\Phi_n\circ K^{(n)}_{\mathrm{cut}}\circ\Phi_n^{-1}
			&=
			\frac12\sum_{a,b\ge1}(a+b)\,p_a p_b\,\partial_{p_{a+b}},
			\label{eq:cut-operator}
			\\[4pt]
			\Phi_n\circ K^{(n)}_{\mathrm{join}}\circ\Phi_n^{-1}
			&=
			\frac12\sum_{a,b\ge1}ab\,p_{a+b}\,\partial_{p_a}\partial_{p_b}.
			\label{eq:join-operator}
		\end{align}
	\end{theorem}
	
	\begin{proof}
		We argue combinatorially in terms of cycle structures of permutations.
		
		\medskip\noindent
		\textbf{Cut operator.}
		Consider a permutation $\sigma\in S_n$ that contains a cycle of length
		$a+b$,
		\[
		(c_1\,c_2\,\dots,c_{a+b}),
		\]
		and suppose that all other cycles of $\sigma$ are kept fixed.  Left
		multiplication by a transposition $(ij)$ may \emph{cut} this
		$(a+b)$–cycle into two cycles of lengths $a$ and $b$.  We first count,
		for a fixed cycle of length $a+b$, how many transpositions produce a
		given ordered pair of cycle lengths $(a,b)$.
		
		Choose on the given $(a+b)$–cycle a point where we cut it.  Cutting at
		a given position produces an ordered pair of cycles: one of length $a$
		and one of length $b$.  There are exactly $a+b$ distinct choices of the
		cut position along the cycle, hence $a+b$ transpositions that realize
		such a cut.
		
		In the algebra of symmetric functions, the factor $p_{a+b}$ records the
		presence of a cycle of length $a+b$.  Acting on $p_\lambda$,
		$\partial_{p_{a+b}}$ ``extracts'' such a cycle and lowers its
		multiplicity by one.  The factor $p_a p_b$ then records the two
		resulting cycles of lengths $a$ and $b$.  Thus, ignoring for the moment
		any symmetry factors, the contribution of ``cutting a cycle of
		length $a+b$ into cycles of lengths $a$ and $b$'' is encoded by the
		operator
		\[
		(a+b)\,p_a p_b\,\partial_{p_{a+b}}.
		\]
		
		Next, we sum over all pairs $(a,b)\in\mathbb{Z}_{\ge1}^2$.  In the sum
		\(
		\sum_{a,b\ge1}(\cdots)
		\)
		the pairs $(a,b)$ and $(b,a)$ describe the same cutting operation: they
		both correspond to cutting a cycle of length $a+b$ into two cycles of
		lengths $\{a,b\}$ without ordering.  Therefore each unordered pair
		$\{a,b\}$ is counted twice, and we must divide by~$2$ to correct for
		this overcounting.  This yields
		\[
		\Phi_n\circ K^{(n)}_{\mathrm{cut}}\circ\Phi_n^{-1}
		=
		\frac12\sum_{a,b\ge1}(a+b)\,p_a p_b\,\partial_{p_{a+b}},
		\]
		which is \eqref{eq:cut-operator}.
		
		\medskip\noindent
		\textbf{Join operator.}
		Now we consider transpositions that join two cycles into a longer
		one.  Let $\sigma$ contain two disjoint cycles
		\[
		(c_1\,\dots,c_a),\qquad (d_1\,\dots,d_b)
		\]
		of lengths $a$ and $b$, respectively.  A transposition $(ij)$ with
		$i$ on the first cycle and $j$ on the second one joins these two cycles
		into a single cycle of length $a+b$.
		
		For a fixed pair of cycles of lengths $a$ and $b$, the number of
		transpositions that join them is $ab$: we may choose any of the $a$
		points on the first cycle and any of the $b$ points on the second
		cycle, so there are $a\cdot b$ possible transpositions.
		
		In the power sum language, a cycle of length $a$ contributes a factor
		$p_a$, and a cycle of length $b$ contributes $p_b$.  Thus
		\(
		\partial_{p_a}\partial_{p_b}
		\)
		removes a pair of cycles of lengths $a$ and $b$, while $p_{a+b}$
		inserts the joined cycle of length $a+b$.  Ignoring again symmetry
		factors, the local operation of joining two cycles of lengths $a$ and
		$b$ is encoded by
		\[
		ab\,p_{a+b}\,\partial_{p_a}\partial_{p_b}.
		\]
		
		Summing over all positive integers $a,b$ double counts the unordered
		pair $\{a,b\}$ (the pair $(a,b)$ and $(b,a)$ describe the same joining
		operation), so we must again divide by $2$.  Therefore
		\[
		\Phi_n\circ K^{(n)}_{\mathrm{join}}\circ\Phi_n^{-1}
		=
		\frac12\sum_{a,b\ge1}ab\,p_{a+b}\,\partial_{p_a}\partial_{p_b},
		\]
		which is exactly \eqref{eq:join-operator}.  This completes the proof.
	\end{proof}

	\begin{corollary}[cut-join operator]\label{cor:W0}
		The normalized transposition class sum $ K^{(n)}[2]$ corresponds to the classical cut-join operator
		\begin{equation}\label{eq:W0-classical}
			W_{[2]}=	\Phi_n\!\circ \big( K^{(n)}_{[2]}\big)\circ \Phi_n^{-1}
			=\frac12\sum_{a,b\ge1}\Big((a{+}b)\,p_a p_b\,\pp_{p_{a+b}}+ab\,p_{a+b}\,\pp_{p_a}\pp_{p_b}\Big).
		\end{equation}
	\end{corollary}
	Then we can draw the process as follows:
	\[
	\begin{tikzcd}
		Z(\C[S_n]) \arrow[r,"K^{(n)}_{[2]}"] \arrow[d,"\Phi_n"'] &
		Z(\C[S_n]) \arrow[d,"\Phi_n"] \\
		\Lambda^{(n)} \arrow[r,"W_{[2]}"'] &
		\ \Lambda^{(n)} \quad .
	\end{tikzcd}\]
	The commutativity of the square is precisely Corollary~1.2, which
	identifies the action of the transposition class sum $K^{(n)}_{[2]}$ with the
	cut-join operator $W_{[2]}$ via the cycle index map.
	
	To facilitate the subsequent introduction of background charges, we rewrite the cut-join operator into the following form:
	\begin{eqnarray}\label{eq:W0-classical}
		W_{[2]}&=&	\Phi_n\!\circ \big( K^{(n)}[2]\big)\circ \Phi_n^{-1}
		\nonumber
		\\
		&=&\frac12\Big(\sum_{a,b\ge1}(a{+}b)\,p_a p_b\,\pp_{p_{a+b}}-\sum_{k\ge1}(k-1)k\,p_k\,\partial_{p_k}\Big)
		\nonumber
		\\
		&+&
		\frac12\Big(\sum_{a,b\ge1}ab\,p_{a+b}\,\pp_{p_a}\pp_{p_b}+
		\sum_{k\ge1}(k-1)k\,p_k\,\partial_{p_k}
		\Big).
	\end{eqnarray}

	\begin{remark}[Diagonal part and stability]
		The operators in Corollary~\ref{cor:W0} should be understood as the
		stable limit of the family $\{K^{(n)}_{\mathrm{cut}},K^{(n)}_{\mathrm{join}}\}_{n\geq1}$, in the sense
		that for fixed $a,b$ the corresponding differential operators coincide for all sufficiently
		large $n$.  In particular, the cut-join operator $W_{[2]}$
		contains no diagonal part: there is no term proportional to
		$\sum_k k p_k\partial_{p_k}$.  When we later introduce an explicit diagonal operator
		$D=\sum_k (k-1)k\,p_k\partial_{p_k}$ and write $W_{[2]} = C+J+D$, the splitting
		\[
		C \mapsto C-\tfrac12D,\qquad J \mapsto J+\tfrac12D
		\]
		is purely algebraic and does \emph{not} correspond to an additional combinatorial
		cut/join rule; its main purpose is to isolate the diagonal part, which will be interpreted
		as a background charge in the $\beta$–ensemble and equivariant geometric settings.
	\end{remark}

	Let $\alpha_{-k}=p_k$, $\alpha_k = k\,\partial_{p_k}$ for $k>0$, so that
	$[\alpha_m,\alpha_n]=m\,\delta_{m+n,0}$.  Then we can introduce Heisenberg realization of cut-join operator, i.e. the cubic
	Heisenberg operator  
	\begin{equation}\label{eq:W0-Heis}
		W_{[2]}=
		\frac{1}{2}\sum_{i,j\ge1}
		\bigl(:\alpha_{-i}\alpha_{-j}\alpha_{i+j}:
		+:\alpha_{-(i+j)}\alpha_i\alpha_j:\bigr).
	\end{equation}

	
	\section{Ladder construction and Jucys--Murphy lifting}\label{sec:ladder}
	
	\subsection{Layered block algebra and the raising map}
	Let $Z_n:=Z(\C[S_n])$ be the center (equivalently, the class algebra) of the group algebra.
	Denote by
	\[
	\ell_n:\C[S_n]\hookrightarrow \C[S_{n+1}]
	\]
	the standard embedding that fixes $n+1$ as a $1$-cycle.
	Let
	\[
	\pi_{n+1}:\C[S_{n+1}]\to Z_{n+1},\qquad
	\pi_{n+1}(x):=\frac{1}{(n+1)!}\sum_{g\in S_{n+1}} g\,x\,g^{-1}
	\]
	be the central projection. We define the \emph{centered raising map}
	\[
	r_n:=r_n^{\mathrm{cen}}:=\pi_{n+1}\circ \ell_n\;:\; Z_n\to Z_{n+1}.
	\]
	
	\paragraph{Layered block algebra.}
	Consider the layered block algebra
	\[
	\mathcal L^{\mathrm{cen}}
	=\bigoplus_{n\ge 0}\Bigl(\End(Z_n)\,e_{n,n}\ \oplus\ \Hom(Z_n,Z_{n+1})\,e_{n+1,n}\Bigr),
	\]
	with the obvious matrix like multiplication.
	
	\paragraph{The transposition class and $w_0$.}
	Let
	\[
	K^{(n)}_{[2]}:=\sum_{1\le i<j\le n} (ij)\ \in Z_n
	\]
	be the sum of all transpositions (a central element). Set
	\[
	w_0^{(n)}:=\frac12 K^{(n)}_{[2]}\in Z_n,\qquad
	w_0:=\sum_{n\ge0} L_{w_0^{(n)}}\,e_{n,n},\qquad
	p:=\sum_{n\ge0} r_n\,e_{n+1,n},
	\]
	where $L_{w_0^{(n)}}$ denotes left multiplication by $w_0^{(n)}$ on $Z_n$.
	
	Define the (centered) ladder operator
	\begin{equation}\label{eq:E_cen_def}
		E:=[w_0,p]
		=\sum_{n\ge0}\Bigl(L_{w_0^{(n+1)}}\circ r_n - r_n\circ L_{w_0^{(n)}}\Bigr)\,e_{n+1,n}
		\ \in\ \mathcal L^{\mathrm{cen}}.
	\end{equation}
	
	\paragraph{Characteristic map.}
	Let $\Lambda=\C[p_1,p_2,\dots]$ be the algebra of symmetric functions in power sum variables,
	and $\Lambda^{(n)}$ its degree-$n$ part.
	For each partition $\lambda\vdash n$, let
	\[
	K_\lambda:=\sum_{\sigma\in S_n:\,\mathrm{type}(\sigma)=\lambda}\sigma \ \in Z_n
	\]
	be the class sum. Define the characteristic map (no normalization)
	\begin{equation}\label{eq:E_cha_def}
		\Phi_n: Z_n\xrightarrow{\ \sim\ }\Lambda^{(n)},\qquad \Phi_n(K_\lambda):=p_\lambda:=\prod_i p_{\lambda_i}.
	\end{equation}
	
	\begin{lemma}[Intertwining of the centered ladder]\label{thm:centered_ladder}
		Let 
		\[
		E_1:=[W_{[2]},p_1]=\sum_{m\ge1} m\,p_{m+1}\,\partial_{p_m}.
		\]
		Then for all $z\in Z_n$ we have
		$
		\Phi_{n+1}(E z)=E_1\bigl(\Phi_n(z)\bigr).
		$
		Equivalently, $\Phi_{n+1}\circ E = E_1\circ \Phi_n$ as maps $Z_n\to \Lambda^{(n+1)}$.
	\end{lemma}
	
	\begin{proof}
		It suffices to check on the class sum basis $\{K_\lambda\}_{\lambda\vdash n}$.
		
		First note that $p_n(K_\lambda)=K_{\lambda\cup(1)}$, because $\ell_n$ fixes $n+1$ and the central projection
		$\pi_{n+1}$ replaces an element by its conjugacy class sum.
		
		Next observe that the commutator in \eqref{eq:E_cen_def} cancels the contributions of transpositions entirely inside
		$\{1,\dots,n\}$, leaving only those involving the new letter $n+1$. Concretely,
		\[
		K^{(n+1)}- \ell_n\bigl(K^{(n)}\bigr)=\sum_{i=1}^n (i,n+1)=:J_{n+1},
		\]
		the Jucys-Murphy element. Acting on a permutation of cycle type $\lambda$, each transposition $(i,n+1)$
		inserts $n+1$ into the cycle containing $i$.
		If that cycle has length $m$, there are exactly $m$ possible insertion slots, hence a factor $m$.
		Therefore we obtains the explicit class sum formula
		\[
		E(K_\lambda)=\sum_{m\ge1} m\,m_m(\lambda)\,K_{\lambda^{(m\to m+1)}},
		\]
		where $m_m(\lambda)$ is the multiplicity of the part $m$ in $\lambda$, and
		$\lambda^{(m\to m+1)}$ is the partition of $n+1$ obtained from $\lambda$ by replacing one part $m$ by $m+1$.
		
		Applying $\Phi_{n+1}$ yields
		\[
		\Phi_{n+1}(E K_\lambda)
		=\sum_{m\ge1} m\,m_m(\lambda)\,p_{\lambda^{(m\to m+1)}}
		=\sum_{m\ge1} m\,p_{m+1}\,\partial_{p_m}(p_\lambda)
		=E_1(p_\lambda),
		\]
		as claimed.
		
		Then, the exchange diagrams for JM lifting can be  drawn as follows.
		\begin{center}
			\begin{tikzcd}[column sep=54pt,row sep=18pt]
				Z_n \arrow[r,"{E=[w_0,p]}"] \arrow[d,"{\Phi_n}"'] & Z_{n+1} \arrow[d,"{\Phi_{n+1}}"] \\
				\Lambda^{(n)} \arrow[r,"{E_1}"'] & \Lambda^{(n+1)} \quad .
			\end{tikzcd}
		\end{center}
	\end{proof}

	\subsection{Jucys-Murphy elements and lifting}\label{subsec:JM_centered} 
	
	Recall the $(n+1)$-st Jucys-Murphy element
	\[
	J_{n+1}:=\sum_{i=1}^{n}(i,n+1)\in \C[S_{n+1}].
	\]
	Use the definitions of class algebra (the center) $Z_n:=Z(\C[S_n])$, the standard embedding $	\ell_n$
	and the central projection $\pi_{n+1}$, we consider the following constructions.
	
	\paragraph{Centered lifting operator.}
	Define the centered Jucys-Murphy lifting operator
	\begin{equation}\label{eq:Tn_cen_def}
		T_n^{\mathrm{cen}}:Z_n\to Z_{n+1},\qquad
		T_n^{\mathrm{cen}}(z):=\pi_{n+1}\!\bigl(J_{n+1}\,\ell_n(z)\bigr).
	\end{equation}
	This modification is essential to maintain the characteristic map as an isomorphism
	$Z_n\simeq \Lambda^{(n)}$.

	Since we have know the definition of characteristic map from~\eqref{eq:E_cha_def}, the Jucys-Murphy lifting operator can be constructed.
	\begin{lemma}[Jucys-Murphy lifting and the raising operator]\label{thm:JM_lifting_centered}
		Let
		\[
		E_1:=\sum_{m\ge1} m\,p_{m+1}\,\partial_{p_m}:\Lambda^{(n)}\to \Lambda^{(n+1)}.
		\]
		Then for all $z\in Z_n$ we have the intertwining identity
		\begin{equation}\label{eq:JM_intertwining}
			\Phi_{n+1}\!\bigl(T_n^{\mathrm{cen}}(z)\bigr)\;=\;E_1\bigl(\Phi_n(z)\bigr).
		\end{equation}
		Equivalently, the diagram
		\[
		\begin{tikzcd}
			Z_n \arrow[r,"T_n^{\mathrm{cen}}"] \arrow[d,"\Phi_n"'] &
			Z_{n+1} \arrow[d,"\Phi_{n+1}"] \\
			\Lambda^{(n)} \arrow[r,"E_1"'] &
			\Lambda^{(n+1)}
		\end{tikzcd}
		\]
		commutes.
	\end{lemma}
	
	\begin{proof}
		It suffices to check \eqref{eq:JM_intertwining} on the class-sum basis $\{K_\lambda\}_{\lambda\vdash n}$.
		
		Fix $\lambda\vdash n$ and consider a permutation $\sigma$ of cycle type $\lambda$.
		Multiplying by a transposition $(i,n+1)$ inserts the new letter $n+1$ into the cycle of $\sigma$ that contains $i$.
		If that cycle has length $m$, then there are exactly $m$ insertion slots (equivalently, $m$ choices of $i$ along that
		cycle) producing a permutation whose cycle type is obtained by replacing one part $m$ by $m+1$.
		Summing over all $i=1,\dots,n$ and then centralizing yields
		\[
		T_n^{\mathrm{cen}}(K_\lambda)
		=\sum_{m\ge1} m\,m_m(\lambda)\,K_{\lambda^{(m\to m+1)}},
		\]
		where $m_m(\lambda)$ is the multiplicity of the part $m$ in $\lambda$, and
		$\lambda^{(m\to m+1)}$ denotes the partition of $n+1$ obtained from $\lambda$ by replacing one part $m$ by $m+1$.
		Applying $\Phi_{n+1}$ gives
		\[
		\Phi_{n+1}\!\bigl(T_n^{\mathrm{cen}}(K_\lambda)\bigr)
		=\sum_{m\ge1} m\,m_m(\lambda)\,p_{\lambda^{(m\to m+1)}}.
		\]
		On the other hand,
		\[
		E_1(p_\lambda)
		=\sum_{m\ge1} m\,p_{m+1}\,\partial_{p_m}(p_\lambda)
		=\sum_{m\ge1} m\,m_m(\lambda)\,p_{\lambda^{(m\to m+1)}},
		\]
		which matches the previous expression and proves \eqref{eq:JM_intertwining}.
	\end{proof}

	\subsection{Gaussian $\beta$-ensemble, Virasoro and the cubic $W_0^{(\beta)}$  }\label{sec:beta}
	
	In this section we consider the  Gaussian Hermitian $\beta$-ensemble
	\begin{equation}
		Z_G^{(\beta)}(t)
		=
		\int_{\mathbb{R}^N}
		\prod_{i=1}^N dx_i\,
		\exp\Bigl(
		-\frac{1}{2}\sum_{i=1}^N x_i^2
		+\sum_{k\ge1} t_k\sum_{i=1}^N x_i^k
		\Bigr)\,
		\prod_{1\le i<j\le N} |x_i-x_j|^{2\beta},
		\label{eq:Gaussian-beta-partition}
	\end{equation}
	with $t_k$ related to power sums by $t_k = p_k/k$. This is the standard Hermitian
	$\beta$-ensemble \cite{Dyson1962,Mehta}, whose Virasoro constraints
	are well known in the case $\beta=1$ and can be generalized to arbitrary $\beta>0$,
	see e.g.\ \cite{ChekhovEynardOrantinBeta}.

	\begin{lemma}[Virasoro constraints for the Gaussian $\beta$-ensemble\cite{ChekhovEynardOrantinBeta}]
		\label{lem:Gaussian-beta-Virasoro}
		For each $n\ge-1$, the partition function~\eqref{eq:Gaussian-beta-partition}
		satisfies the Virasoro constraints
		\begin{equation}
			L_n^{(\beta)}\,Z_G^{(\beta)}(t)=0,\qquad n\ge-1,
			\label{eq:Gaussian-beta-Virasoro}
		\end{equation}
		where
		\begin{equation}
			L_n^{(\beta)}
			=
			\sum_{k\ge1}k\,t_k\,\frac{\partial}{\partial t_{k+n}}
			+\beta\sum_{a+b=n}\frac{\partial^2}{\partial t_a\partial t_b}
			+\bigl(1-\beta\bigr)(n+1)\frac{\partial}{\partial t_n}
			+2\beta N\,\frac{\partial}{\partial t_n}
			\label{eq:Ln-beta-t}
		\end{equation}
		and $L_{-1}^{(\beta)}$ is obtained from the same formula with the obvious
		restriction on summation.
	\end{lemma}
	Note that the partition function $Z_G^{(\beta)}(p)$ should be viewed as a prototypical $\tau$-function-like object whose Ward identities encode an infinite family of constraints.
	For $\beta=1$ these constraints generate the standard matrix model Virasoro algebra, while for general $\beta$ they deform the Heisenberg background charge and naturally organize into a $W_{1+\infty}$-type structure.

	Let power sums $p_k=kt_k$ and writing the Virasoro generators as differential operators in the variables $p_k$ we obtain:
	
	\begin{lemma}[Cubic $W$-constraint]
		\label{lem:Gaussian-beta-W}
		Define the cubic operator
		\begin{equation}
			W_0^{(\beta)}
			=
			\frac{1}{2}\sum_{n\ge1} n\,p_n\,L_n^{(\beta)},
			\label{eq:W0beta-def-from-L}
		\end{equation}
		where $L_n^{(\beta)}$ is given by~\eqref{eq:Ln-beta-t} expressed in $p$-variables.
		Then $W_0^{(\beta)}$ can be written explicitly as
		\begin{equation}
			W_0^{(\beta)}
			=
			\frac12\sum_{a,b\ge1}
			\Bigl(
			\beta(a+b)\,p_a p_b\,\partial_{p_{a+b}}
			+ab\,p_{a+b}\,\partial_{p_a}\partial_{p_b}
			\Bigr)
			+\frac12\sum_{k\ge1}
			\bigl((1-\beta)(k-1)+2\beta N\bigr)k\,p_k\partial_{p_k}.
			\label{eq:W0beta-explicit}
		\end{equation}
		Moreover,
		\begin{equation}
			W_0^{(\beta)}\,Z_G^{(\beta)}(p)=0.
		\end{equation}
	\end{lemma}
	
	\begin{proof}
		We expand~\eqref{eq:W0beta-def-from-L} term by term. The quadratic part of
		$L_n^{(\beta)}$ gives the sums over $(a,b)$ with coefficients $(a+b)$ and
		$ab$, which reconstruct the cut/join operators $C,J$ in the sense of
		\cite{MironovMorozovNatanzon09}. The linear part in $\partial_{p_n}$ produces
		two diagonal operators,
		\[
		D:=\sum_{k\ge1}(k-1)k\,p_k\partial_{p_k},\qquad
		E:=\sum_{k\ge1}k\,p_k\partial_{p_k},
		\]
		with coefficients $(1-\beta)/2$ and $\beta N$, respectively. Then from (\ref{eq:W0-classical}), this yields
		\begin{equation}\label{eq:w-beta-linear}
			W_0^{(\beta)}
			=\beta(C-\tfrac12D)+(J+\tfrac12D)+\beta N\,E
			=
			\beta C + J + \frac{1-\beta}{2}\,D + \beta N\,E,
		\end{equation}
		which is equivalent to~\eqref{eq:W0beta-explicit}. Since $W_0^{(\beta)}$ is a
		fixed linear combination of the Virasoro generators and each
		$L_n^{(\beta)}$ annihilates $Z_G^{(\beta)}$ by
		Lemma~\ref{lem:Gaussian-beta-Virasoro}, we obtain
		$W_0^{(\beta)}Z_G^{(\beta)}=0$.
	\end{proof}
	
	The cubic combination in Lemma~\ref{lem:Gaussian-beta-W} is singled out because it is precisely the generating operator appearing in $W$-representations/superintegrability for $\beta$-deformed hierarchies, and it is the operator we can geometrically realize explicitly.
	The  $W_0^{(\beta)}$ considered here can be derived directly from the Ward identities
	of the Gaussian $\beta$-ensemble.  Closely related $\beta$-deformed
	$W$-representations and superintegrable partition function hierarchies
	were constructed in~\cite{wlzz22}, where $W$-operators are
	used as generating operators for large families of $(\beta\text{-deformed})$
	matrix models.
	It is convenient to rewrite the diagonal part in terms of a background charge.
	Introduce Heisenberg modes
	\[
	J_{-k}=p_k,\qquad
	J_k=\beta k\,\partial_{p_k},\qquad
	J_0=\beta N,\qquad
	[J_k,J_l]=\beta k\,\delta_{k+l,0},
	\]
	and rescale $a_k:=J_k/\sqrt{\beta}$, so that $[a_k,a_l]=k\delta_{k+l,0}$.
	
	\begin{proposition}[Background charge parametrization]
		\label{prop:background-charge}
		Let $b=\sqrt{\beta}$ and define
		\begin{equation}
			Q_b:=\frac{b-\frac{1}{b}}{2}.
			\label{eq:Qb-def}
		\end{equation}
		Then the diagonal part of $W_0^{(\beta)}$ can be written as
		\begin{equation}
			\frac12\sum_{k\ge1}
			\bigl((1-\beta)(k-1)+2\beta N\bigr)k\,p_k\partial_{p_k}
			=
			-\sqrt{\beta}\,Q_b\,D + \beta N\,E,
			\label{eq:diag-in-Qb}
		\end{equation}
		with $D,E$ as above. In particular,
		\begin{equation}
			\frac{1-\beta}{2}
			=-\sqrt{\beta}\,Q_b.
		\end{equation}
	\end{proposition}
	
	\begin{proof}
		By definition,
		\(
		Q_b=(b-b^{-1})/2
		\),
		so
		\(
		-\sqrt{\beta}Q_b = -bQ_b = (1-b^2)/2 = (1-\beta)/2
		\).
		Substituting this into the coefficient in front of $D$ gives
		\eqref{eq:diag-in-Qb}.
	\end{proof}
	In free-field realizations of 2d CFT, the improved stress tensor
	\[
	T(z)=\frac12:(\partial\phi)^2:(z)+Q_b\,\partial^2\phi(z)
	\]
	has central charge \(c=1-12Q_b^2\), see e.g.\ \cite{DiFrancescoMathCFT}.
	The extra $Q_b$-term gives rise to the
	diagonal operator $D$ in $W_0^{(\beta)}$.
	Thus Proposition~\ref{prop:background-charge} identifies the diagonal correction term in \(W_0^{(\beta)}\) with the standard background charge improvement.
	Under the $\Omega$-background parametrization $\beta=\varepsilon_1/\varepsilon_2$, the self-dual point \(\varepsilon_1=\varepsilon_2\) corresponds to \(Q_b=0\), recovering the undeformed (self-dual) CFT normalization.

	
	In the $\Omega$-background parametrization
	\(
	\beta=\varepsilon_1/\varepsilon_2
	\),
	we have
	\[
	Q_b
	=\frac{\sqrt{\varepsilon_1/\varepsilon_2}
		-\sqrt{\varepsilon_2/\varepsilon_1}}{2}
	=\frac{\varepsilon_1-\varepsilon_2}{2\sqrt{\varepsilon_1\varepsilon_2}},
	\]
	In particular, the extra diagonal term disappears if and only if
	$\varepsilon_1=\varepsilon_2$ (self-dual $\Omega$-background).
	
	
	\section{ Geometric realization on the Hilbert schemes}\label{sec:hilb}
	
	\subsection{Nakajima-Grojnowski Fock space}
	
	Let $\Hilb^n(\C^2)$ be the Hilbert scheme of $n$ points on the plane, with
	$T=(\C^\times)^2$ acting by scaling the coordinates of $\C^2$. Following
	Grojnowski and Nakajima \cite{Groj,NakajimaLectures}, we consider the
	$T$–equivariant cohomology
	\[
	\FHilb \;:=\; \bigoplus_{n\ge0} H_T^\ast\!\bigl(\Hilb^n(\C^2)\bigr).
	\]
	We regard $ \Fock_{\Hilb}$ as a graded space, where the grading is given by $n$.
	
	\paragraph{Bosonic Fock space and grading.}
	Let $\Lambda=\C[p_1,p_2,\dots]$ be the ring of symmetric functions in the power sums, graded by total degree
	\[
	\Lambda=\bigoplus_{n\ge 0}\Lambda^{(n)},\qquad \deg(p_k)=k.
	\]
	We will work with the direct sum grading (no completion is needed for the statements below):
	each $\Lambda^{(n)}$ is finite-dimensional over $\C$.
	
	\paragraph{Heisenberg operators and normalization convention.}
	On the geometric side, Grojnowski-Nakajima construct operators
	\[
	\alpha_{-k}: \Fock_{\Hilb}\to  \Fock_{\Hilb}\quad (k\ge 1),
	\qquad
	\alpha_{k}: \Fock_{\Hilb}\to  \Fock_{\Hilb}\quad (k\ge 1),
	\]
	(the Nakajima creation/annihilation operators) satisfying the Heisenberg commutation relations
	\[
	[\alpha_m,\alpha_n]=m\,\delta_{m+n,0}\,\langle 1,1\rangle_X,
	\qquad m,n\in\Z,
	\]
	where $\langle\cdot,\cdot\rangle_X$ denotes the natural $T$-equivariant pairing on $H_T^*(X)$.
		For $X=\C^2$ we have
	\[
	\langle 1,1\rangle_X=\int_X 1=\frac{1}{\varepsilon_1\varepsilon_2},
	\]
	hence $[\alpha_m,\alpha_n]= m\,\delta_{m+n,0}/(\varepsilon_1\varepsilon_2)$.
	Let $|0\rangle\in H_T^*(\Hilb^0(X))\cong H_T^*(\mathrm{pt})$ be the vacuum vector, characterized by
	\[
	\alpha_k|0\rangle=0\qquad (k>0).
	\]

	On the algebraic side, we identify the bosonic Heisenberg operators on $\Lambda$ by the convention
	\[
	\alpha_{-k}\ \leftrightarrow\ k\,p_k,\qquad
	\alpha_k\ \leftrightarrow\ \frac{1}{\varepsilon_1\varepsilon_2}\,\frac{\partial}{\partial p_k}\qquad (k\ge1).
	\]
	Equivalently, multiplication by $p_k$ corresponds to $\frac1k\alpha_{-k}$ under the Fock identification.

	On the algebraic side, we identify the bosonic Heisenberg operators on $\Lambda$ by the convention
	\begin{equation}\label{eq:Heis_normalization_Lambda}
		\alpha_{-k}\ \leftrightarrow\ k\,p_k,\qquad
		\alpha_{k}\ \leftrightarrow\ \partial_{p_k}\qquad (k\ge 1).
	\end{equation}
	Equivalently, multiplication by $p_k$ corresponds to $\frac1k\,\alpha_{-k}$ under the Fock identification.
	
	\begin{theorem}[Nakajima-Grojnowski identification, graded form]\label{thm:Phi_Hilb_graded}
		There exists a unique graded linear isomorphism
		\[
		\Phi_{\Hilb}:\Lambda \xrightarrow{\ \sim\ }\Fock_{\Hilb},
		\qquad
		\Phi_{\Hilb}=\bigoplus_{n\ge 0}\Phi_{\Hilb,n},\ \ \Phi_{\Hilb,n}:\Lambda^{(n)}\xrightarrow{\sim} H_T^*(\Hilb^n(X)),
		\]
		such that
		\[
		\Phi_{\Hilb}(1)=|0\rangle \in H_T^*(\Hilb^0(X))
		\quad\text{and}\quad
		\Phi_{\Hilb}\!\bigl(p_k\,f\bigr)=\frac1k\,\alpha_{-k}\,\Phi_{\Hilb}(f)\qquad (k\ge 1,\ f\in\Lambda).
		\]
		Equivalently, $\Phi_{\Hilb}$ intertwines the Heisenberg actions with the normalization \eqref{eq:Heis_normalization_Lambda}:
		\[
		\Phi_{\Hilb}\circ (k\,p_k)=\alpha_{-k}\circ \Phi_{\Hilb},\qquad
		\Phi_{\Hilb}\circ \partial_{p_k}=\alpha_{k}\circ \Phi_{\Hilb}\qquad (k\ge 1).
		\]
	\end{theorem}
	
	\begin{remark}[Note that $\Phi_{\Hilb}$ is invertible although both sides are infinite-dimensional]
		The map is graded: each $\Phi_{\Hilb,n}$ is a linear isomorphism between the finite-dimensional pieces
		$\Lambda^{(n)}$ and $H_T^*(\Hilb^n(X))$ (after localizing in the equivariant parameters, if desired).
		Hence $\Phi_{\Hilb}=\bigoplus_n \Phi_{\Hilb,n}$ is an isomorphism of graded direct sums, with inverse
		$\Phi_{\Hilb}^{-1}=\bigoplus_n \Phi_{\Hilb,n}^{-1}$.
	\end{remark}
	
	\paragraph{Centered characteristic map and the combined identification.}\quad
	
	Because $\Phi_{\Hilb}$ preserves the grading $\deg(p_k)=k$, it restricts to an isomorphism
	\[
	\Phi_{\Hilb}^{(n)}:\Lambda^{(n)}\xrightarrow{\sim} H_T^*(\Hilb^n(\C^2)).
	\]
	Composing with $\Phi_n:Z_n\to \Lambda^{(n)}$ yields a canonical graded identification
	\[
	\Psi_n:=\Phi_{\Hilb}^{(n)}\circ \Phi_n:\ Z_n\ \xrightarrow{\ \sim\ }\ H_T^*(\Hilb^n(\C^2)).
	\]
	In other words, the following diagram commutes for each $n$:
	\begin{center}
		\begin{tikzpicture}[
			node distance=28mm,
			every node/.style={font=\normalsize},
			arr/.style={-Latex, thick}
			]
			\node (Z) {$Z(\C[S_n])$};
			\node (L) [right=of Z] {$\Lambda^{(n)}$};
			\node (H) [right=of L] {$H_T^*(\Hilb^n(\C^2))$};
			
			\draw[arr] (Z) -- node[above] {$\Phi_n$} (L);
			\draw[arr] (L) -- node[above] {$\Phi_{\Hilb}^{(n)}$} (H);
			\draw[arr, bend left=35] (Z) to node[below] {$\Psi_n$} (H);
		\end{tikzpicture}
	\end{center}

	\paragraph{Compatibility with raising/creation operators.}
	Let $E_1=\sum_{m\ge 1} m\,p_{m+1}\partial_{p_m}$ be the raising operator on $\Lambda$
	as in Section~\ref{subsec:JM_centered}. Under the normalization \eqref{eq:Heis_normalization_Lambda},
	multiplication by $p_1$ corresponds to $\alpha_{-1}$:
	\[
	\Phi_{\Hilb}\circ p_1 = \alpha_{-1}\circ \Phi_{\Hilb}.
	\]
	Thus, whenever a centered lifting operator $T_n^{\mathrm{cen}}:Z(\C[S_n])\to Z(\C[S_{n+1}])$
	intertwines with $E_1$ via $\Phi_{n+1}\circ T_n^{\mathrm{cen}}=E_1\circ \Phi_n$
	(Section~\ref{subsec:JM_centered}), we obtain the combined commutative diagram
	\[
	\begin{tikzcd}[column sep=11em]
		Z(\C[S_n]) \arrow[r,"T_n^{\mathrm{cen}}"] \arrow[d,"\Psi_n"'] &
		Z(\C[S_{n+1}]) \arrow[d,"\Psi_{n+1}"] \\
		H_T^*(\Hilb^n(\C^2)) \arrow[r,"\;\;\text{(correspondence induced by }E_1\text{)}\;\;"'] &
		H_T^*(\Hilb^{n+1}(\C^2)),
	\end{tikzcd}
	\]
	where the bottom arrow is the geometric operator corresponding to $E_1$ under $\Phi_{\Hilb}$
	(and can be expressed in terms of Nakajima correspondences / Heisenberg generators).

	\subsection{Hecke correspondences and the operators $E_1$ }
	
	We now discuss the geometric realization of the raising operator
	\[
	E_1 = \sum_{m\ge1} m\,p_{m+1}\,\frac{\partial}{\partial p_m}
	\]
	and of the cubic operator $W_0$ in terms of Hecke correspondences on Hilbert
	schemes.
	
	\paragraph{The Hecke correspondence $Z_{n,n+1}$}
	
	Let $X=\C^2$ and write $\Hilb^n(X)$ for the Hilbert scheme of $n$ points on $X$.
	Let
	\[
	Z_{n,n+1}\subset \Hilb^n(X)\times \Hilb^{n+1}(X)
	\]
	be the \emph{Hecke correspondence} parametrizing nested pairs of ideals $(I,J)$ such that
	\[
	J\subset I\subset \C[x,y],\qquad \dim_\C(\C[x,y]/J)=n+1,\quad \dim_\C(\C[x,y]/I)=n.
	\]
	Equivalently, there is an exact sequence
	\[
	0\longrightarrow J\longrightarrow I\longrightarrow \C_x\longrightarrow 0
	\]
	for some $x\in X$, and $Z_{n,n+1}$ remembers the nested pair $(I,J)$ but forgets $x$.
	Let
	\[
	\pi_1,\pi_2:\ Z_{n,n+1}\longrightarrow \Hilb^n(X),\ \Hilb^{n+1}(X)
	\]
	be the two projections. The quotient $I/J$ defines a line bundle $L$ on $Z_{n,n+1}$ whose fiber is the
	one-dimensional vector space $I/J$.

	\begin{definition}[Geometric $E_1$]\label{def:E1_geo}
		Define the geometric operator
		\[
		E_1^{\mathrm{geo}}:\ H_T^*(\Hilb^n(X))\longrightarrow H_T^*(\Hilb^{n+1}(X)),\qquad
		E_1^{\mathrm{geo}}(\gamma):=(\pi_2)_*\bigl(\pi_1^*\gamma\ \cup\ c_1(\mathcal L)\bigr),
		\]
		where $T=(\C^\times)^2$ acts on $X=\C^2$ by scaling the coordinates.
	\end{definition}
	We consider the (equivariant) Fock space $\Fock_{\Hilb}$,
	then on the algebraic side let $\Lambda=\C[p_1,p_2,\dots]$ be the ring of symmetric functions in power sums and
	write $\Lambda^{(n)}$ for the degree-$n$ part.
	
	\paragraph{Heisenberg normalization on $\Lambda$.}
	We use the standard bosonic realization of the Heisenberg algebra on $\Lambda$:
	for $k\ge 1$ set
	\begin{equation}\label{eq:Heis_norm}
		a_{-k}:=p_k,\qquad a_k:=k\,\frac{\partial}{\partial p_k},
	\end{equation}
	so that $[a_m,a_{-n}]=m\,\delta_{m,n}$ for $m,n\ge 1$.

	\paragraph{Nakajima operators.}
	Grojnowski and Nakajima construct incidence correspondences $Z^{(k)}\subset \Hilb^n(X)\times \Hilb^{n+k}(X)$,
	functorial in $n$ and $k>0$, giving rise to creation/annihilation operators
	\[
	q_{-k}(\gamma):H_T^*(\Hilb^n)\to H_T^*(\Hilb^{n+k}),\qquad
	q_{k}(\gamma):H_T^*(\Hilb^{n+k})\to H_T^*(\Hilb^{n}),
	\]
	that satisfy Heisenberg commutation relations. For $X=\C^2$ and $\gamma=[\C^2]$ we denote the resulting
	Heisenberg generators by $\{a_k\}_{k\in\Z\setminus\{0\}}$ and fix the identification with \eqref{eq:Heis_norm}.
	For convenient, we let $\Hilb^n:=\Hilb^n(\mathbb{C}^2)$ be the Hilbert scheme of $n$ points on $\mathbb{C}^2$.

	\begin{theorem}[Geometric realization of $E_1$]\label{thm:E1_geo}
		Under the Fock identification $\Phi_{\Hilb}$, the geometric operator $E_1^{\mathrm{geo}}$ from
		Definition~\ref{def:E1_geo} corresponds to the standard ladder operator on $\Lambda$:
		\[
		\Phi_{\Hilb}^{-1}\circ E_1^{\mathrm{geo}}\circ \Phi_{\Hilb}
		\;=\;
		\sum_{m\ge 1} m\,p_{m+1}\,\frac{\partial}{\partial p_m}.
		\]
		Equivalently, $E_1^{\mathrm{geo}}$ is represented on $\Fock_{\Hilb}$ by the quadratic Heisenberg expression
		\begin{equation}\label{eq:E1_Heis}
			E_1^{\mathrm{geo}}
			\;=\;\sum_{m\ge 1} a_{-(m+1)}\,a_m.
		\end{equation}
	\end{theorem}
	
	\begin{proof}
		We recall two standard ingredients.
		
		\smallskip\noindent
		(i) \emph{Nakajima operators and the Heisenberg action.}
		The incidence correspondences $Z^{(k)}$ produce creation/annihilation operators that satisfy the Heisenberg relations.
		Under $\Phi_{\Hilb}$ we identify these operators with the bosonic generators \eqref{eq:Heis_norm}.
		
		\smallskip\noindent
		(ii) \emph{The Hecke correspondence as a quadratic Heisenberg operator.}
		The correspondence $Z_{n,n+1}$ is a special case of Nakajima's incidence correspondences. The cohomological insertion
		$c_1(L)$ yields a normally ordered quadratic operator of total degree $+1$, namely \eqref{eq:E1_Heis}.
		(With the convention \eqref{eq:Heis_norm}, inserting an extra prefactor $m$ in front of $a_{-(m+1)}a_m$ would give
		$m^2 p_{m+1}\partial_{p_m}$ and hence would be incompatible with the desired ladder operator.)
		
		Transport \eqref{eq:E1_Heis} to $\Lambda$ via \eqref{eq:Heis_norm}$:$
		\[
		\Phi_{\Hilb}^{-1}\circ E_1^{\mathrm{geo}}\circ \Phi_{\Hilb}
		=\sum_{m\ge 1} a_{-(m+1)}a_m
		=\sum_{m\ge 1} p_{m+1}\cdot\Bigl(m\frac{\partial}{\partial p_m}\Bigr)
		=\sum_{m\ge 1} m\,p_{m+1}\frac{\partial}{\partial p_m}.
		\]
	\end{proof}

	\subsection{Geometric realization of $W_0^{(\beta)}$ on Hilbert schemes $\Hilb^n(\C^2)$}
	\label{subsec:Hilb-W0beta}
	
	\paragraph{Convolution of kernels}\cite{NakajimaLectures,NakajimaAnnals,MaulikOkounkov12,CGKhomology}
	Let $A\subset X\times Y$ and $B\subset Y\times Z$ be correspondences. Denote by $[A]\in H_T^*(X\times Y)$ and
	$[B]\in H_T^*(Y\times Z)$ their equivariant cohomology classes. Define the convolution kernel
	\begin{equation}\label{eq:conv2}
		[B]\star [A] := (\pi_{13})_*\bigl(\pi_{12}^*[A]\cup \pi_{23}^*[B]\bigr)\in H_T^*(X\times Z),
	\end{equation}
	where $\pi_{ab}$ are the projections from $X\times Y\times Z$ to $X_a\times X_b$.
	For three correspondences $A\subset X_1\times X_2$, $B\subset X_2\times X_3$, $C\subset X_3\times X_4$, define
	\begin{equation}\label{eq:conv3}
		[C]\star [B]\star [A]
		:=(\pi_{14})_*\bigl(\pi_{12}^*[A]\cup \pi_{23}^*[B]\cup \pi_{34}^*[C]\bigr)\in H_T^*(X_1\times X_4),
	\end{equation}
	with $\pi_{ab}$ the projections from $X_1\times X_2\times X_3\times X_4$to $X_a\times X_b$.
	By definition of correspondence composition, the operator induced by the kernel
	\eqref{eq:conv2} (resp.\ \eqref{eq:conv3}) is the composite of the operators induced by
	$A,B$ (resp.\ $A,B,C$).


	As we have known the definition of $W_0^{(\beta)}$ from section~\ref{sec:beta}, now we consider the geometric realization of it based on the above setup.
	\begin{theorem}[Geometric realization of $W_0^{(\beta)}$ on $\Hilb^n(\C^2)$]\label{thm:W0beta}
		Define the cubic operator on $ \Lambda=\C[p_1,p_2,\dots]$ by
		\begin{equation}\label{eq:W0beta}
			W_0^{(\beta)}
			=\frac{\beta}{2}\sum_{i,j\ge 1}\NO{a_{-i}a_{-j}a_{i+j}}
			+\frac{1}{2}\sum_{i,j\ge 1}\NO{a_{-(i+j)}a_i a_j}
			+\frac12\sum_{k\ge 1}\Bigl((1-\beta)(k-1)+2\beta N\Bigr)a_{-k}a_k,
		\end{equation}
		where $a_{\pm k}$ are the Heisenberg operators (identified with \eqref{eq:Heis_norm}) and
		$\NO{\cdot}$ denotes normal ordering with respect to the Fock grading.
		Let $W_{0,\mathrm{geo}}^{(\beta)}$ be the operator on $\FHilb$ defined by the same
		normally ordered expression \eqref{eq:W0beta}, but interpreted as a correspondence operator
		on $\Hilb^n(\C^2)$ obtained by convolving Nakajima incidence correspondences (and for the
		diagonal term, by cup product with the canonical equivariant class that produces the
		normal ordering correction).
		Then
		\begin{equation}\label{eq:intertwine}
			\Phi_{\Hilb}\circ W_0^{(\beta)}\;=\;W_{0,\mathrm{geo}}^{(\beta)}\circ \Phi_{\Hilb}.
		\end{equation}
		In particular, $W_0^{(\beta)}$ admits a geometric realization on $\Hilb^n(\C^2)$ as a
		push-pull operator associated with triple incidence correspondences, together with the
		canonical equivariant class responsible for the diagonal/background-charge contribution.
	\end{theorem}
	
	Before proceeding to the proof, we explain how the Heisenberg expression \eqref{eq:W0beta} corresponds to the formula in the work of Maulik and Okounkov \cite{MaulikOkounkov12}. In their framework, the operator of quantum multiplication decomposes into a cubic part and a quadratic part. The quadratic contribution involves a coefficient of the form
	$-\Bigl(a_i+\frac{1-n}{2}(t_1+t_2)\Bigr)$
	multiplying $\alpha^{(i)}_{-n}(1)\alpha^{(i)}_{n}(\mathrm{pt})$.  Specializing to the rank-one case and using the parameter identification
	\[
	t_1=1,\qquad t_2=-\beta,\qquad a_i\ \leftrightarrow\ -\beta N,
	\]
	together with the standard identification of the Nakajima operators with the Heisenberg modes $a_{\pm k}$ on the Fock space $\Lambda=\mathbb{C}[p_1,p_2,\dots]$, we obtain for $k=n>0$ the coefficient
	\[
	-\Bigl(-\beta N+\frac{1-k}{2}(1-\beta)\Bigr)
	=\beta N+\frac{k-1}{2}(1-\beta)
	=\frac12\Bigl((1-\beta)(k-1)+2\beta N\Bigr).
	\]
    This coincides exactly with the coefficient of the third (quadratic) term in \eqref{eq:W0beta}. Consequently, the operator $W_0^{(\beta)}$ in \eqref{eq:W0beta} is precisely the Maulik-Okounkov operator, i.e., the sum of the cubic and quadratic components under the above specialization. 
    In particular, this identifies $W_0^{(\beta)}$ as the geometric correspondence operator on the Hilbert scheme,defined by the push-pull along the diagonal $\Delta\subset \mathrm{Hilb}^n\times \mathrm{Hilb}^n$ with the insertion of the class $c_1(V_{\mathrm{taut}}^\ast)$, thereby providing the claimed geometric realization.

	The proof has two steps: (i) any normally ordered polynomial in Heisenberg generators has a
	geometric realization by convolution of incidence correspondences; (ii) the particular cubic
	combination \eqref{eq:W0beta} is identified with the operator on $ \Lambda=\C[p_1,p_2,\dots]$ via
	\eqref{eq:Heis_norm}, and the diagonal term matches the canonical equivariant class (the same term analyzed in Proposition~\ref{prop:background-charge}).
	
	\paragraph{(i): correspondence realization of normally ordered Heisenberg monomials}
	
	\begin{lemma}[Nakajima incidence correspondences realize the Heisenberg algebra]
		There exist correspondences $Z^{(k)}\subset \Hilb^n(\C^2)\times \Hilb^{n+k}(\C^2)$, functorial in
		$n$ and $k>0$, such that the induced operators on $\FHilb$ satisfy the Heisenberg relations.
		Under $\Phi_{\Hilb}$ the correspondence operators coincide with the explicit generators
		\eqref{eq:Heis_norm}.
	\end{lemma}
	
	\begin{proof}
		This is the foundational construction of Nakajima \cite{NakajimaAnnals}, refined equivariantly
		for $S=\C^2$. The identification with power sums $p_k$ and $k\,\partial_{p_k}$ is the
		standard Fock space normalization and is assumed as our convention \eqref{eq:Heis_norm}.
	\end{proof}
	
	\begin{lemma}[Convolution realizes normal ordered monomials]\label{lem:convolution}
		For any normally ordered monomial $\NO{a_{k_1}\cdots a_{k_r}}$ (with all creation operators
		$a_{-m}$ placed to the left of all annihilation operators $a_{m}$), the operator on $\FHilb$
		equals a push-pull operator obtained by convolving the corresponding Nakajima incidence
		correspondences $Z^{(|k_\ell|)}$ in the natural order.
		In particular, each cubic term $\NO{a_{-i}a_{-j}a_{i+j}}$ and $\NO{a_{-(i+j)}a_i a_j}$ is
		realized by a triple incidence correspondence.
	\end{lemma}
	
	\begin{proof}
		Each $a_{\pm k}$ is defined by a push--pull operator along $Z^{(k)}$; the product of two
		operators corresponds to convolution of correspondences, and normal ordering ensures that
		intermediate Hilbert-scheme degrees are well-defined in the Fock direct sum.
		Iterating this gives the claim for any normally ordered monomial; the cubic case is the
		specialization $r=3$.
	\end{proof}
	
	\paragraph{(ii): identification under $\Phi_{\Hilb}$ and the diagonal term}
	
	\begin{lemma}[Intertwining for the cubic cut/join parts]\label{lem:cutjoin}
		Let
		\[
		W_{\mathrm{cubic}}^{(\beta)}:=\frac{\beta}{2}\sum_{i,j\ge 1}\NO{a_{-i}a_{-j}a_{i+j}}
		+\frac{1}{2}\sum_{i,j\ge 1}\NO{a_{-(i+j)}a_i a_j}.
		\]
		Then $\Phi_{\Hilb}\circ W_{\mathrm{cubic}}^{(\beta)}=
		W_{\mathrm{cubic,geo}}^{(\beta)}\circ \Phi_{\Hilb}$, where
		$W_{\mathrm{cubic,geo}}^{(\beta)}$ is the same expression interpreted geometrically as in
		Lemma~\ref{lem:convolution}.
	\end{lemma}
	
	\begin{proof}
		By Lemma~\ref{lem:convolution}, the geometric operator
		$W_{\mathrm{cubic,geo}}^{(\beta)}$ is, by definition, the normally ordered expression in the
		geometric Heisenberg generators $a_k$.
		Under $\Phi_{\Hilb}$ the generators are identified with \eqref{eq:Heis_norm}, so the same
		normally ordered expression defines an operator on $ \Lambda=\C[p_1,p_2,\dots]$.
		Hence the two actions coincide and are intertwined by $\Phi_{\Hilb}$.
	\end{proof}

	\begin{lemma}[Diagonal/background-charge term]\label{lem:diag}
		Fix parameters $\beta,N\in\C$ and define on $\Lambda$ the quadratic operator
		\begin{equation}\label{eq:Wdiag_def}
			W^{(\beta)}_{\mathrm{diag}}
			:=\frac12\sum_{k\ge 1}\Bigl((1-\beta)(k-1)+2\beta N\Bigr)\,a_{-k}a_k,
		\end{equation}
		where $a_{\pm k}$ act by \eqref{eq:Heis_norm}.
		Define $W^{(\beta)}_{\mathrm{diag,geo}}$ on $\Fock_{\Hilb}$ by the \emph{same} formula \eqref{eq:Wdiag_def}, but with
		$a_{\pm k}$ interpreted as the geometric Nakajima operators.
		Then
		\begin{equation}\label{eq:diag_intertwine}
			\Phi_{\Hilb}\circ W^{(\beta)}_{\mathrm{diag}}
			\;=\;
			W^{(\beta)}_{\mathrm{diag,geo}}\circ \Phi_{\Hilb}.
		\end{equation}
		Moreover, on the Nakajima/Heisenberg monomial basis
		\[
		p_\lambda:=\prod_{i}p_{\lambda_i}\in \Lambda^{(|\lambda|)}
		\qquad\Longleftrightarrow\qquad
		a_{-\lambda}|0\rangle:=\prod_i a_{-\lambda_i}|0\rangle\in H_T^*(\Hilb^{|\lambda|}),
		\]
		the operator $W^{(\beta)}_{\mathrm{diag}}$ (and hence $W^{(\beta)}_{\mathrm{diag,geo}}$) acts diagonally with
		eigenvalue
		\begin{equation}\label{eq:diag_evalue}
			W^{(\beta)}_{\mathrm{diag}}(p_\lambda)
			=
			\frac12\sum_{k\ge 1}\Bigl((1-\beta)(k-1)+2\beta N\Bigr)\,k\,m_k(\lambda)\,p_\lambda,
		\end{equation}
		where $m_k(\lambda)$ is the multiplicity of the part $k$ in $\lambda$.
	\end{lemma}
	
	\begin{proof}
		\emph{Step 1: algebraic well-definedness and diagonal action.}
		Under \eqref{eq:Heis_norm}, we have
		\[
		a_{-k}a_k = p_k\cdot k\partial_{p_k},
		\]
		so $W^{(\beta)}_{\mathrm{diag}}$ is a well-defined differential operator on $\Lambda$.
		On the monomial $p_\lambda=\prod_i p_{\lambda_i}$ we compute
		\[
		k\partial_{p_k}(p_\lambda)=k\,m_k(\lambda)\,p_\lambda/p_k,
		\qquad\Rightarrow\qquad
		a_{-k}a_k(p_\lambda)=p_k\cdot k\,m_k(\lambda)\,p_\lambda/p_k=k\,m_k(\lambda)\,p_\lambda.
		\]
		Substituting into \eqref{eq:Wdiag_def} gives the claimed eigenvalue formula \eqref{eq:diag_evalue}.
		
		\smallskip
		\emph{Step 2: geometric meaning of $a_{-k}a_k$.}
		Geometrically, $a_{-k}$ is the push-pull operator defined by the incidence correspondence
		$Z^{(k)}\subset \Hilb^n\times \Hilb^{n+k}$, and $a_k$ is defined by the transpose correspondence.
		Hence $a_{-k}a_k$ is a degree-preserving correspondence operator on $\Fock_{\Hilb}$ obtained by composing
		$Z^{(k)}$ with $(Z^{(k)})^t$ (equivalently, by convolution of kernels).
		
		\smallskip
		\emph{Step 3: intertwining.}
		By construction, $\Phi_{\Hilb}$ is an isomorphism of Heisenberg modules, i.e.\ it intertwines every generator
		$a_{\pm k}$ for $k\ge 1$:
		\[
		\Phi_{\Hilb}\circ a_{\pm k}=a_{\pm k}\circ \Phi_{\Hilb}.
		\]
		Therefore $\Phi_{\Hilb}$ also intertwines any polynomial in these generators.
		In particular, it intertwines the quadratic combination \eqref{eq:Wdiag_def}, which is exactly
		\eqref{eq:diag_intertwine}.
		Finally, since $\Phi_{\Hilb}(p_\lambda)=a_{-\lambda}|0\rangle$ under this normalization, the diagonal action
		\eqref{eq:diag_evalue} transports verbatim to $\Fock_{\Hilb}$.
	\end{proof}

	\begin{proof}[Proof of Theorem~\ref{thm:W0beta}]
		Write $W_0^{(\beta)}=W_{\mathrm{cubic}}^{(\beta)}+W_{\mathrm{diag}}^{(\beta)}$ as in
		Lemmas~\ref{lem:cutjoin} and \ref{lem:diag}.
		Each part is intertwined by $\Phi_{\Hilb}$ with its geometric realization; hence their sum is
		intertwined, giving \eqref{eq:intertwine}.
		The ``push-pull along triple incidence correspondences'' description follows from
		Lemma~\ref{lem:convolution}.
	\end{proof}

	\begin{remark}[On the role of localization]
		A complete proof can be given by computing matrix elements in the $T$-fixed point basis
		$\{|\lambda\rangle\}$ (indexed by partitions $\lambda$). 
		We compute the contributions from the triple incidence correspondences and the canonical equivariant class, and verify that the resulting matrix elements agree with those of the normally ordered operator
		\eqref{eq:W0beta} expressed in Heisenberg modes.
		Such fixed-point computations use explicit
		formulas for tangent and tautological weights on $\Hilb^n(\C^2)$ and are standard in the
		literature around \cite{NakajimaAnnals}. The quiver variety analogs and Yangian actions are
		developed, for example, in \cite{VaragnoloNakajima}.
	\end{remark}

	
	\begin{theorem}[Geometric realization of $W_{[2]}$]\label{thm:geom-W0}
		Let
		\[
		W_{[2]}
		=\frac12\sum_{i,j\ge 1}\Bigl(a_{-i}a_{-j}a_{i+j} + a_{-(i+j)}a_i a_j\Bigr)
		\]
		be the (normally ordered) cubic cut-and-join operator on $\Lam$ (equivalently on $\Fock$).
		Then $W_{[2]}$ is realized on $\FHilb$ by a correspondence kernel
		\[
		\mathcal W_{[2]}\ \subset\ \bigsqcup_{n\ge 0}\Hilb^n(\C^2)\times \Hilb^n(\C^2),
		\]
		defined as a disjoint union of two geometric branches (cut/join) coming from triple
		convolutions of the basic Hecke correspondences $Z^{(k)}$ and their transposes.
		Under $\Phi_{\Hilb}$ this geometric action agrees with the cut-and-join action on
		symmetric functions.
	\end{theorem}
	
	\begin{proof}
		Fix $n\ge 0$ and $i,j\ge 1$ with $n\ge i+j$.
		Consider the two length-$3$ correspondence chains that preserve $n$:
		
		\medskip\noindent
		\textbf{(i) The cut part (corresponding to $a_{-i}a_{-j}a_{i+j}$).}
		Set
		\[
		X_1=\Hilb^n,\quad
		X_2=\Hilb^{n-(i+j)},\quad
		X_3=\Hilb^{n-i},\quad
		X_4=\Hilb^n.
		\]
		Let
		\[
		A:=Z^{(-(i+j))}\subset X_1\times X_2,\qquad
		B:=Z^{(j)}\subset X_2\times X_3,\qquad
		C:=Z^{(i)}\subset X_3\times X_4,
		\]
		where $Z^{(-(i+j))}$ is the transpose of $Z^{(i+j)}\subset \Hilb_{n-(i+j)}\times \Hilb^n$.
		Define the \emph{cut correspondence space} as the fiber-product locus
		\begin{equation}\label{eq:Wcut-space}
			\mathcal W^{\mathrm{cut}}_{i,j}(n)
			\;:=\;
			\pi_{12}^{-1}(A)\ \cap\ \pi_{23}^{-1}(B)\ \cap\ \pi_{34}^{-1}(C)
			\ \subset\ X_1\times X_2\times X_3\times X_4,
		\end{equation}
		and its induced kernel on $\Hilb^n\times \Hilb^n$ by
		\begin{equation}\label{eq:Wcut-kernel}
			K^{\mathrm{cut}}_{i,j}(n)
			\;:=\;
			(\pi_{14})_*\bigl[\mathcal W^{\mathrm{cut}}_{i,j}(n)\bigr]
			\ \in\ H_T^*(\Hilb^n\times \Hilb^n).
		\end{equation}
		By the convolution definition \eqref{eq:conv3}, the correspondence operator induced by
		$K^{\mathrm{cut}}_{i,j}(n)$ is exactly the composite operator
		\[
		H_T^*(\Hilb^n)\xrightarrow{\,a_{i+j}\,}H_T^*(\Hilb^{n-(i+j)})
		\xrightarrow{\,a_{-j}\,}H_T^*(\Hilb^{n-i})
		\xrightarrow{\,a_{-i}\,}H_T^*(\Hilb^n),
		\]
		i.e.\ it equals $a_{-i}a_{-j}a_{i+j}$ on $H_T^*(\Hilb^n)$.
		
		\medskip\noindent
		\textbf{(ii) The join part (corresponding to $a_{-(i+j)}a_i a_j$).}
		Set
		\[
		Y_1=\Hilb^n,\quad
		Y_2=\Hilb^{n-j},\quad
		Y_3=\Hilb^{n-(i+j)},\quad
		Y_4=\Hilb^n.
		\]
		Let
		\[
		A':=Z^{(-j)}\subset Y_1\times Y_2,\qquad
		B':=Z^{(-i)}\subset Y_2\times Y_3,\qquad
		C':=Z^{(i+j)}\subset Y_3\times Y_4.
		\]
		Define the \emph{join correspondence space}
		\begin{equation}\label{eq:Wjoin-space}
			\mathcal W^{\mathrm{join}}_{i,j}(n)
			\;:=\;
			\pi_{12}^{-1}(A')\ \cap\ \pi_{23}^{-1}(B')\ \cap\ \pi_{34}^{-1}(C')
			\ \subset\ Y_1\times Y_2\times Y_3\times Y_4,
		\end{equation}
		and the induced kernel
		\begin{equation}\label{eq:Wjoin-kernel}
			K^{\mathrm{join}}_{i,j}(n)
			\;:=\;
			(\pi_{14})_*\bigl[\mathcal W^{\mathrm{join}}_{i,j}(n)\bigr]
			\ \in\ H_T^*(\Hilb^n\times \Hilb^n).
		\end{equation}
		Again by \eqref{eq:conv3}, the operator induced by $K^{\mathrm{join}}_{i,j}(n)$ equals the
		composite
		\[
		H_T^*(\Hilb^n)\xrightarrow{\,a_{j}\,}H_T^*(\Hilb^{n-j})
		\xrightarrow{\,a_{i}\,}H_T^*(\Hilb^{n-(i+j)})
		\xrightarrow{\,a_{-(i+j)}\,}H_T^*(\Hilb^n),
		\]
		i.e.\ it equals $a_{-(i+j)}a_i a_j$ on $H_T^*(\Hilb^n)$.
		
		\medskip
		\textbf{Normal ordering and the two branches.}
		The two kernels \eqref{eq:Wcut-kernel} and \eqref{eq:Wjoin-kernel} are precisely the
		normally ordered cubic contributions: they are the degree-preserving triple convolutions in
		which the net degree drops first (by an annihilation) and is then restored (by creations),
		or vice versa, matching the monomials $a_{-i}a_{-j}a_{i+j}$ and $a_{-(i+j)}a_i a_j$.
		Other orderings would produce nonnormalordered contributions; discarding them is exactly
		the effect of normal ordering.
		
		Hence, for each $n$ we define the geometric kernel
		\begin{equation}\label{eq:W2kernel}
			K_{[2]}(n)
			\;:=\;\frac12\sum_{i,j\ge 1}\Bigl(K^{\mathrm{cut}}_{i,j}(n)+K^{\mathrm{join}}_{i,j}(n)\Bigr)
			\ \in\ H_T^*(\Hilb^n\times \Hilb^n),
		\end{equation}
		and set $\mathcal W_{[2]}:=\bigsqcup_{n\ge 0}\mathrm{supp}(K_{[2]}(n))\subset\bigsqcup_n\Hilb^n\times\Hilb^n$.
		
		\medskip
		\textbf{Identification with the cut-and-join operator on $\Lam$.}
		Under the Fock identification $\Phi_{\Hilb}$ and the normalization
		$a_{-k}=p_k$, $a_k=k\partial_{p_k}$, the operator induced by \eqref{eq:W2kernel} acts on
		$\Lam$ as
		\[
		\frac12\sum_{i,j\ge 1}\Bigl((i+j)\,p_ip_j\,\frac{\partial}{\partial p_{i+j}}
		+\;ij\,p_{i+j}\frac{\partial^2}{\partial p_i\,\partial p_j}\Bigr),
		\]
		which is the standard cut-and-join operator. Therefore $\mathcal W_{[2]}$ realizes
		$W_{[2]}$ on $\FHilb$, and this action is intertwined with the symmetric function action by
		$\Phi_{\Hilb}$.
		
		\medskip
		\textbf{Geometric meaning of cut-and-join operator.}
		Localizing to $T$-fixed points labelled by partitions $\lambda$, the basic correspondences
		$Z^{(\pm k)}$ add/remove rim hooks of length $k$; hence the cut branch
		$\mathcal W^{\mathrm{cut}}_{i,j}(n)$ realizes splitting a rim hook of length $i+j$ into two
		of lengths $i$ and $j$, while the join branch $\mathcal W^{\mathrm{join}}_{i,j}(n)$ realizes
		gluing two rim hooks of lengths $i$ and $j$ into one of length $i+j$. This matches the two
		summands in $W_{[2]}$.
	\end{proof}
	
	
	\subsection{Exchange diagrams and rim hooks}

	Theorem~\ref{thm:geom-W0}  can be viewed as a precise geometric lift of
	the group side picture. The transposition class
	sum $K_{[2]}^{(n)}\in\C[S_n]$ corresponds, under the cycle index map and the
	Fock identification, to the cubic operator $W_0$ acting on symmetric functions. 
	The cut and join of cycles thus become, on the Hilbert scheme side, explicit operations on
	rim hooks in Young diagrams labelling $T$-fixed points, providing a concrete
	bridge between the combinatorics of permutations and the geometry of Nakajima varieties.

	\begin{example}[The case $n=4$]\label{ex:n4-matrix}
		For $n=4$ the $T$-fixed points on $\Hilb^4(\C^2)$ are labelled by the five
		partitions
		\[
		(4),\ (3,1),\ (2,2),\ (2,1,1),\ (1,1,1,1).
		\]
		Localizing the geometric operator $W_{[2]}$ of Theorem~\ref{thm:geom-W0} to
		these fixed points we obtain a $5\times5$ matrix whose off-diagonal entries
		are in bijection with rim hook cut and join moves between the corresponding
		Young diagrams. The entry $\langle\mu|W_{[2]}|\lambda\rangle$ is non-zero if and
		only if $\mu$ is obtained from $\lambda$ by cutting or joining a rim hook, and
		the matrix element is given by the combinatorial weight $(a+b)$ in the cut
		case and $ab$ in the join case, where $a$ and $b$ denote the lengths of the
		two pieces; see the following figure.
		\begin{figure}[h]
			\centering
			\includegraphics[width=0.7\linewidth]{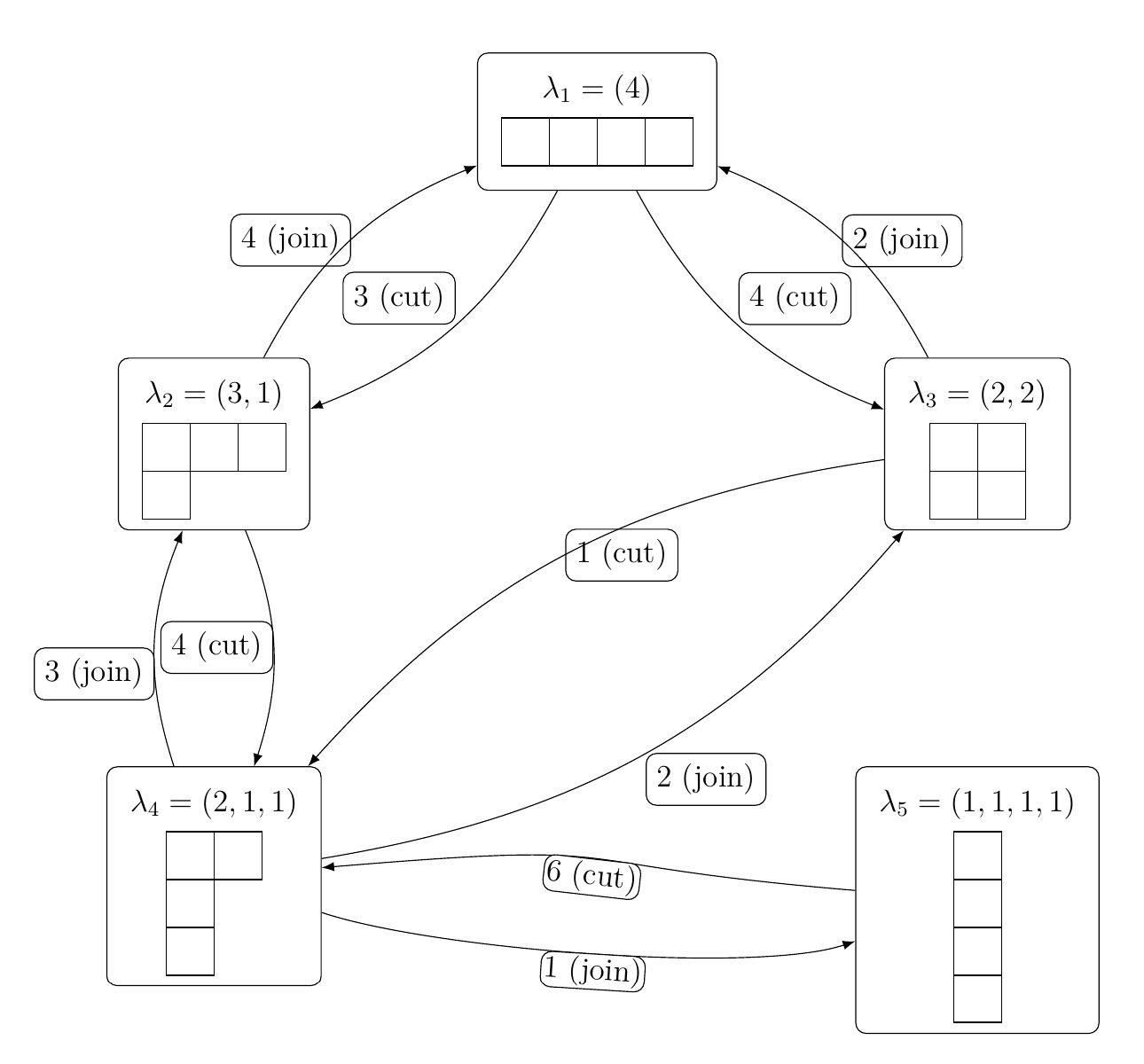}
			\caption[]{$W_{[2]} $ acts on the Young diagram basis in the $n=4$ case.}
			\label{fig:n4}
		\end{figure}
		
		Here, we should noted that: 
		\begin{enumerate}
			\item  each vertex is a partition $\lambda$, together with its Young diagram; 
			\item  each directed edge represents a non-diagonal matrix element $M^{(4)}_{ij}$ of $W_{[2]}$ in the fixed-point basis 
			(Since $W_{[2]}$ preserves the degree $n$, then on the set $\{\,v_\lambda : \lambda \vdash n\,\}$, it defines a matrix
			$M=(M_{\mu,\lambda})$
			such that
			$	W_{[2]}(v_\lambda)=\sum_{\mu\vdash n} M_{\mu,\lambda}\,v_\mu $, where $v_\lambda=\frac{p_\lambda}{z_\lambda}$ and $z_\lambda=\prod_{i\ge 1}i^{m_i} m_i!$);
			\item the number on an edge is the value of that matrix element, and the label $\text{(cut)}/\text{(join)}$ in parentheses indicates whether it comes from the term
			$(a+b)p_a p_b\partial_{p_{a+b}} $ or 
			$ ab\,p_{a+b}\partial_{p_a}\partial_{p_b}$,
			that is, whether it corresponds to “cutting one cycle into two” (cut) or “joining two cycles into one” (join).
		\end{enumerate}
		
		Explicitly, in the ordered basis
		\(|(4)\rangle, |(3,1)\rangle, \dots\) the matrix of $W_{[2]}$ is
		\[
		M^{(4)}=
		\begin{pmatrix}
			0 & 3 & 4 & 0 & 0\\
			4 & 0 & 0 & 4 & 0\\
			2 & 0 & 0 & 1 & 0\\
			0 & 3 & 2 & 0 & 6\\
			0 & 0 & 0 & 1 & 0
		\end{pmatrix}.
		\]
		where each non-diagonal entry is annotated by the corresponding $(a+b)$ or
		$ab$ coming from the rim hook move. This agrees with the symmetric function
		formula of correlary~\ref{cor:W0} specialized to degree-$4$.
	\end{example}
	
	


	\subsection{Geometric the hierarchy operator generated by $W_{[2]}$}

	We rewrite the  hierarchy operator generated from a cubic
	cut-and-join operator $W^{(0)}:=W_{[2]}$ on the bosonic Fock space.  Using the
	explicit formulas
	\[
	E_1=[W_{[2]},p_1]
	=\sum_{n\ge1}n\,p_{n+1}\frac{\partial}{\partial p_n},\]
	\[
	W^{(1)}:=[W_{[2]},E_1]
	=\sum_{k,l\ge1}(k+l-1)p_k p_l
	\frac{\partial}{\partial p_{k+l-1}}
	+
	\sum_{k,l\ge1}(kl)p_{k+l+1}
	\frac{\partial^2}{\partial p_k\partial_{p_l}},
	\]
	we analyze the first two levels
	of the commutator tower
	\(
	W^{(n)}:=\operatorname{ad}_{W_1}^{n-1}(E_1) \ (n\ge2)
	\)
	on the Fock/Hilbert side.  A key point, correcting a common
	misstatement, is that $W_1$ has degree $+1$ (it raises the particle
	number by one), so $W^{(n)}$ raises the number of points by $n$.
	Geometrically, $W^{(n)}$ should therefore be viewed as a linear
	combination of correspondences
	\(
	\Hilb^N\to\Hilb^{N+n}
	\)
	built from one Hecke step and $n$ applications of a cubic incidence
	correspondence.

	With the fact that $
	\deg(p_k)=k, \deg\!\Bigl(\frac{\partial}{\partial p_k}\Bigr)=-k,
	$
	we know that an operator $D$ has \emph{degree} $r$ if it maps the
	homogeneous component of charge $n$ into charge $n+r$.
	We assume the cut-and-join operator $W^{(0)}$ is of degree
	$0$ (as happens for the standard $W_{[2]}$ operator), so that it does
	not change the total number of points.

	\paragraph{The hierarchy operator $W^{(n)}$}\quad
	
	We now introduce the commutator hierarchy generated by $W_1$ and
	$E_1$.
	\begin{definition}
		Define operators $W^{(n)}$ for $n\ge2$ by
		\begin{equation}
			\label{eq:Wn-def}
			W^{(n)}:=\operatorname{ad}_{W^{(1)}}^{n-1}(E_1)
			=
			\underbrace{[W^{(1)},[W^{(1)},\dots,[W^{(1)},}_{n-1}E_1]\dots]].
		\end{equation}
	\end{definition}
	
	\begin{lemma}
		\label{lem:degree-Wn}
		For all $n\ge1$, we have
		\begin{equation}
			\deg(W^{(n)}) = n.
		\end{equation}
		Equivalently, $W^{(n)}$ maps the charge-$N$ subspace of $\Lambda$ to
		charge $N+n$.
	\end{lemma}
	
	\begin{proof}
		By the definition of degree, we know that if $\deg(A)=r$, $\deg(B)=s$, then
		\(
		\deg([A,B])=\deg(AB-BA)=r+s.
		\)	
		Obviously, (\ref{eq:Wn-def}) can be expressed as $W^{(n)} = [W^{(n-1)},W^{(1)}], n\ge3$.
		Thus
		\begin{equation}
			\deg(W^{(n)}) = \deg(W^{(1)})+\deg(W^{(n-1)})
			= 1 + \deg(W^{(n-1)}),
		\end{equation}
		and by induction
		\(
		\deg(W^{(n)})=n.
		\)
	\end{proof}
	\begin{remark}[What $W^{(n)}$ encodes]
		The identity $\ad_A(B)=AB-BA$ implies that $\ad_A^{m}(B)$ is an alternating sum of length-$(m+1)$ words in $A$ and $B$.
		Thus $W^{(n)}$ is (up to the factorial normalization) an alternating sum of compositions of $E_1$ and $W^{(1)}$
		of total length $n$, which is the algebraic shadow of alternating compositions of correspondences on the Hilbert scheme.
	\end{remark}


	\paragraph{Hilbert schemes: geometric meaning of $E_1$ and $W^{(n)}$}\quad

	\begin{definition}[$E_1$ is the Hecke (nested Hilbert) correspondence]\label{def:nested}
		Let $\Hilb^{n,n+1}(\C^2)\subset \Hilb^n(\C^2)\times \Hilb^{n+1}(\C^2)$ be the nested Hilbert scheme
		parametrizing inclusions of ideals $I\subset I'$ with $\mathrm{length}(I'/I)=1$.
		Let $\pi_1,\pi_2$ be the two projections.
		Define the geometric raising operator
		\begin{equation}\label{eq:E1geo}
			E_{1,\Hilb}:=(\pi_2)_*\circ \pi_1^*:\ H_T^*(\Hilb^n)\longrightarrow H_T^*(\Hilb^{n+1}).
		\end{equation}
	\end{definition}
	
	\begin{remark}[Intertwining statement]
		Under the normalization \eqref{eq:Heis_norm}, we have
		\[
		\Phi_{\Hilb}\circ E_1 = E_{1,\Hilb}\circ \Phi_{\Hilb},
		\]
		i.e.\ the algebraic operator $E_{1}$ matches the geometric Hecke step \eqref{eq:E1geo}.
		This is the precise sense in which ``$E_1$ is the Hecke correspondence''.
	\end{remark}
	
	Then, on the geometric side we have
	\begin{itemize}
		\item the Hecke correspondence $E_{1,\Hilb}:H^*_{T}(\Hilb^n)\to
		H^*_{T}(\Hilb^{n+1})$ corresponding to adding one point;
		
		\item a degree-preserving correspondence
		$W[2]_{\Hilb}:H^*_{T}(\Hilb^n)\to H^*_{T}(\Hilb^n)$ encoding the
		incidence of length-two subschemes (the geometric avatar of the
		cut-and-join operator).
	\end{itemize}
	
	\begin{definition}[Geometric meaning of $W^{(1)}$ and $W^{(n)}$] \label{eq:Wgeo_def}
		Define geometric operators by conjugation:
		\begin{equation}
			W^{(n)}_{\Hilb}:=\Phi_{\Hilb}\circ W^{(n)}\circ \Phi_{\Hilb}^{-1}:
			H_T^*(\Hilb^N)\longrightarrow H_T^*(\Hilb^{N+n}).
		\end{equation}
	\end{definition}
	This definition makes the ``intertwining'' automatic. The nontrivial content is the \emph{geometric realization}
	of $W^{(n)}_{\Hilb}$ as explicit combinations of correspondences.

	\begin{lemma}[Commutator of correspondence operators]\label{lem:comm-correspondence}
		Let $A\subset X\times Y$, $B\subset X\times Y$ be correspondences giving
		operators $A_*,B_*:H^*_{T}(X)\to H^*_{T}(Y)$.  Then the commutator
		\[
		[A_*,B_*] = A_*\circ B_* - B_*\circ A_* : H^*_{T}(X)\to H^*_{T}(Y)
		\]
		is represented by the difference of the two composite correspondences
		$A\circ B$ and $B\circ A$.  More precisely,
		\[
		[A_*,B_*] = (A\circ B)_* - (B\circ A)_*.
		\]
	\end{lemma}
	
	\begin{proof}
		By the discussion above, $A_*\circ B_*=(A\circ B)_*$ and
		$B_*\circ A_*=(B\circ A)_*$.  The claim follows by linearity of the
		correspondence formalism.
	\end{proof}
	
	We apply this to the incidence correspondences on products of Hilbert
	schemes.  For typographical convenience we write
	$E_{1,\Hilb}$ and $W_{[2],\Hilb}$ also for their underlying cycles.
	
	\begin{definition}[Degree-$1$ geometric Hamiltonian]\label{def:W1Hilb}
		We define the degree-$1$ geometric operator
		\begin{equation}\label{eq:W1Hilb-def}
			\WHilbOne := [\,W_{[2],\Hilb},E_{1,\Hilb}\,]
			= W_{[2],\Hilb}\circ E_{1,\Hilb} - E_{1,\Hilb}\circ W_{[2],\Hilb}.
		\end{equation}
	\end{definition}
	
	By Lemma~\ref{lem:comm-correspondence}, $\WHilbOne$ is represented by
	the difference of the two fibre product correspondences encoding
	$W_{[2],\Hilb}\circ E_{1,\Hilb}$ and $E_{1,\Hilb}\circ W_{[2],\Hilb}$.

	\begin{proposition}[Interpretation of $W^{(1)}$ and higher $W^{(n)}$]
		\label{prop:Hilb-Wn}
		For all $n\ge2$, we have
		\begin{equation}
			\Phi_{\Hilb}\circ W^{(n)} = \mathbf{W}^{(n)}\circ\Phi_{\Hilb},
		\end{equation}
		and $\mathbf{W}^{(n)}$ has the following interpretation:
		\begin{enumerate}
			\item $\mathbf{W}^{(2)}=[\mathbf{W}^{(1)},\mathbf{E}_1]$ is the difference
			of two composite correspondences
			\[
			\Hilb^n
			\xrightarrow{\ \mathbf{E}_1\ }\Hilb^{n+1}
			\xrightarrow{\ \mathbf{W}^{(2)}\ }\Hilb^{n+2},
			\qquad
			\Hilb^n
			\xrightarrow{\ \mathbf{W}^{(2)},\ }\Hilb^{n+1}
			\xrightarrow{\ \mathbf{E}_1\ }\Hilb^{n+2}.
			\]
			
			Let $A$ be a correspondence operator $\Hilb^m\to \Hilb^{m+r}$ and $B$ a correspondence operator $\Hilb^m\to\Hilb^{m+s}$.
			Then $[A,B]$ is represented by the difference of the two fiber product correspondences encoding $A\circ B$ and $B\circ A$.
			In particular, $W^{(1)}_{\Hilb}=[W[2]_{\Hilb},E_{1,\Hilb}]$ measures the noncommutativity of the degree-preserving
			cut-and-join correspondence $W[2]_{\Hilb}$ with the Hecke step $E_{1,\Hilb}$.
			\item 
			By Lemma~\ref{lem:degree-Wn}, $W^{(n)}_{\Hilb}$ maps $\Hilb^N$ to $\Hilb^{N+n}$, hence the higher commutators
			\(
			\mathbf{W}^{(n)}=\operatorname{ad}_{\mathbf{W}_1}^{n-1}(\mathbf{E}_1)
			\) is naturally built from
			correspondences on flags of nested ideals of total length $n$. The iterated commutator structure implies an \emph{alternating}
			sum over all ways of interlacing $E_{1,\Hilb}$ with $n-1$ applications of $W^{(1)}_{\Hilb}$ (or, equivalently, $n-1$
			commutators with the degree-$1$ operator $W^{(1)}_{\Hilb}$).
		\end{enumerate}
	\end{proposition}
	
	Interpreting the hierarchy operator as an alternating composition of Hecke steps and cubic incidence correspondences suggests a direct geometric origin for higher Hamiltonians in $\beta$-deformed hierarchies.
	This is the geometric counterpart of building integrable Hamiltonians by iterated commutators in a $W$-algebra, and it is the mechanism expected to persist (in colored form) on quiver varieties relevant for quiver gauge theories.

	\begin{theorem}[Geometric interpretation of $\WHilbn$]\label{thm:WHilb}
		For all $n\ge2$ the operators $\WHilbn$ defined in
		Definiiton \ref{eq:Wgeo_def} satisfy the following properties.
		
		\begin{enumerate}
			\item \emph{Compatibility with the Fock isomorphism.}
			For all $n\ge1$ we have the intertwining relation
			\begin{equation}\label{eq:intertwining-again}
				\Phi_{\Hilb}\circ W^{(n)} = \WHilbn\circ\Phi_{\Hilb}.
			\end{equation}
			
			\item \emph{First non-trivial level $n=2$.}
			The operator $\WHilbTwo$ is the commutator of the degree-$1$
			Hamiltonian with the Hecke step:
			\[
			\WHilbTwo = [\,\WHilbOne,\EHilb\,].
			\]
			Geometrically, $\WHilbTwo$ is represented by the difference of the
			two composite correspondences
			\[
			\Hilb^n \xrightarrow{\ \EHilb\ } \Hilb^{n+1}
			\xrightarrow{\ \WHilbOne\ } \Hilb^{n+2},
			\quad
			\Hilb^n \xrightarrow{\ \WHilbOne\ } \Hilb^{n+1}
			\xrightarrow{\ \EHilb\ } \Hilb^{n+2}.
			\]
			
			\item \emph{Higher levels as iterated commutators.}
			Let $\mathrm{ad}_A(B)=[A,B]$ and
			$\mathrm{ad}_A^{\,k}(B)=\mathrm{ad}_A(\mathrm{ad}_A^{\,k-1}(B))$.  For
			$n\ge2$ we have
			\begin{equation}\label{eq:WHilb-iterated}
				\WHilbn = \mathrm{ad}_{\WHilbOne}^{\,n-1}(\EHilb).
			\end{equation}
			In particular, $\WHilbn$ maps $H^*_{T}(\Hilb^N)$ to
			$H^*_{T}(\Hilb^{N+n})$, and it is represented by an alternating sum
			over all ways of interlacing a single Hecke correspondence
			$E_{1,\Hilb}$ with $(n-1)$ copies of $W[2]_{\Hilb}$ (or, equivalently,
			$(n-1)$ copies of $\WHilbOne$) along flags of nested ideals
			\[
			I_N \supset I_{N+1} \supset \cdots \supset I_{N+n}
			\]
			of total length $n$.
		\end{enumerate}
	\end{theorem}
	
	\begin{proof}
		
		\smallskip\noindent
		\textbf{(i)}  By definition we set
		$\WHilbn=\Phi_{\Hilb}\circ W^{(n)}\circ\Phi_{\Hilb}^{-1}$.  Multiplying
		on the right by $\Phi_{\Hilb}$ gives~\eqref{eq:intertwining-again}
		immediately:
		\[
		\Phi_{\Hilb}\circ W^{(n)}
		=(\Phi_{\Hilb}\circ W^{(n)}\circ \Phi_{\Hilb}^{-1})\circ\Phi_{\Hilb}
		=\WHilbn\circ\Phi_{\Hilb}.
		\]
		
		\smallskip\noindent
		\textbf{(ii)}  On the symmetric function side, we define the degree-$1$
		Hamiltonian by
		\[
		W^{(1)} := [\,W_{[2]},E_1\,],
		\]
		and we set $W^{(2)}=[W^{(1)},E_1]$ as in the usual $W$-algebra hierarchy.
		By Definition~\ref{eq:Wgeo_def},
		\[
		\WHilbOne=\Phi_{\Hilb} W^{(1)}\Phi_{\Hilb}^{-1},\qquad
		\WHilbTwo=\Phi_{\Hilb} W^{(2)}\Phi_{\Hilb}^{-1}.
		\]
		Using functoriality of conjugation and the identity
		$\Phi_{\Hilb}E_1\Phi_{\Hilb}^{-1}=\EHilb$ we obtain
		\begin{align*}
			\WHilbTwo
			&= \Phi_{\Hilb}[\,W^{(1)},E_1\,]\Phi_{\Hilb}^{-1}
			= [\,\Phi_{\Hilb}W^{(1)}\Phi_{\Hilb}^{-1},
			\Phi_{\Hilb}E_1\Phi_{\Hilb}^{-1}\,]   \\
			&= [\,\WHilbOne,\EHilb\,].
		\end{align*}
		This proves the operator identity.  By
		Lemma~\ref{lem:comm-correspondence}, $\WHilbOne$ is represented by the
		difference of the composite correspondences
		$W[2]_{\Hilb}\circ \EHilb$ and $\EHilb\circ W[2]_{\Hilb}$, and similarly
		the commutator $[\WHilbOne,\EHilb]$ is represented by the difference of
		the two composites
		\[
		\EHilb\circ\WHilbOne,\qquad \WHilbOne\circ\EHilb,
		\]
		which correspond exactly to the two diagrams in point~(2) of the
		statement.  This gives the required geometric description.
		
		\smallskip\noindent
		\textbf{(iii)}   On the symmetric function side the higher operators are
		defined by
		\[
		W^{(n)} = \mathrm{ad}_{W^{(1)}}^{\,n-1}(E_1), \qquad n\ge2.
		\]
		Using the compatibility of conjugation with iterated commutators we have
		\begin{align*}
			\WHilbn
			&= \Phi_{\Hilb} W^{(n)}\Phi_{\Hilb}^{-1}
			= \Phi_{\Hilb}\bigl(\mathrm{ad}_{W^{(1)}}^{\,n-1}(E_1)\bigr)\Phi_{\Hilb}^{-1} \\
			&= \mathrm{ad}_{\Phi_{\Hilb}W^{(1)}\Phi_{\Hilb}^{-1}}^{\,n-1}
			(\Phi_{\Hilb}E_1\Phi_{\Hilb}^{-1})
			= \mathrm{ad}_{\WHilbOne}^{\,n-1}(\EHilb),
		\end{align*}
		which proves~\eqref{eq:WHilb-iterated}.
		
		By Lemma~3.14 in the original text (which shows that $W^{(n)}$ shifts
		the grading by $+n$ on the Fock space) and the intertwining
		\eqref{eq:intertwining-again}, $\WHilbn$ maps $H^*_{T}(\Hilb^N)$ to
		$H^*_{T}(\Hilb^{N+n})$, as stated.
		
		Finally, expanding the iterated commutator
		$\mathrm{ad}_{\WHilbOne}^{\,n-1}(\EHilb)$ yields an alternating sum of
		words in the two operators $\WHilbOne$ and $\EHilb$ containing exactly
		one copy of $\EHilb$ and $(n-1)$ copies of $\WHilbOne$.  Each such word
		is a composite of correspondences and is therefore represented by an
		iterated fibre product parametrizing chains of subschemes (flags of
		ideals)
		\[
		I_N \supset I_{N+1}\supset\cdots\supset I_{N+n}
		\]
		of total length $n$.  The sign attached to a given word is determined by
		its position in the commutator expansion.  This gives the asserted
		description of $\WHilbn$ as an alternating sum over correspondences on
		flags of nested ideals, completing the proof.
	\end{proof}

	
	\section*{Conclusion}
	This work provides an explicit dictionary from symmetric group central elements to $W$-operators on symmetric functions and further to
	geometric correspondences on Hilbert schemes.
	On the algebraic side, the transposition class sum produces the Hurwitz cut-and-join operator $W[2]$ together with an intrinsic diagonal term,
	and the ladder/Jucys-Murphy lifting packages the basic raising operator $E_1$ and its commutator tower into a structured hierarchy.
	On the geometric side, the Grojnowski-Nakajima identification transports these operators to
	$\bigoplus_{n\ge0}H_T^*(\Hilb^n(\mathbb C^2))$, where $E_1$ becomes the Hecke correspondence and the diagonal/background charge corrections admit
	precise fixed-point localization formulas in terms of tautological and tangent $T$-weights.
	The $\beta$-deformation is naturally explained by normal ordering in the Heisenberg algebra: diagonal corrections change with $\beta$,
	whereas the leading symbols (and thus the first-order cut/join combinatorics) remain universal.
	
	Several directions are suggested by the resulting ``operator-correspondence'' viewpoint.
	First, we can develop the colored/affine ADE generalization systematically on Nakajima quiver varieties, where infinite dimension vectors contribute and
	the resulting direct sum carries a genuine Fock representation; the expectation is that the hierarchyr mechanism persists (with vertex-wise Hecke steps)
	and produces families of commuting Hamiltonians relevant to quiver gauge theories.
	Second, the explicit diagonal corrections obtained from localization are well suited for comparisons with $\Omega$-background partition functions,
	AGT/W-algebra constraints, and superintegrable $\beta$-ensemble hierarchies.
	Third, the bridge back to Hurwitz theory suggests refined and colored variants of cut-and-join equations, as well as connections to integrable hierarchies and
	quantum (affine/toroidal) algebra actions in cohomology and $K$-theory.

	\section*{Acknowledgement}
	The author would like to thank Anatoli Kirillov, Xinxing Tang, Yehao Zhou and Wei-Zhong Zhao for helpful
	communications and discussions. This work is  supported by the China Postdoctoral Science Foundation(2025M783074).

\end{document}